\documentclass[sigconf, nonacm]{aamas}

\usepackage{wrapfig}
\usepackage{balance}

\usepackage{amsmath}
\usepackage{algorithm}
\usepackage{algpseudocode}
\usepackage{orcidlink}

\usepackage{tikz}
\usetikzlibrary{automata, arrows.meta, shapes.geometric, positioning}

\newtheorem{definition}{Definition}
\newtheorem{theorem}{Theorem}
\newtheorem{proposition}{Proposition}

\newcommand{\bsigma}{\bar{\sigma}}
\newcommand{\btau}{\bar{\tau}}

\renewcommand{\Pr}{\mathbb{P}}

\DeclareMathOperator{\AS}{{\mathit{AS}}}
\DeclareMathOperator{\NZ}{{\mathit{NZ}}}

\def\set#1{\{ #1 \}}
\def\Diam{\lozenge}

\newif\ifdraft\draftfalse
\ifdraft
\def\RM#1{{\textcolor{red}{\textbf{RM:} #1}}}
\def\AL#1{{\textcolor{magenta}{\textbf{AL:} #1}}}
\def\MG#1{{\textcolor{blue}{\textbf{MG:} #1}}}
\else
\def\RM#1{}
\def\AL#1{}
\def\MG#1{}
\fi

\newcommand{\longshort}[2]{#1} 

\newcommand{\OMIT}[1]{}

\newcommand{\Sem}[2]{[\![#1]\!]_{#2}}

\title{Solving Qualitative Multi-Objective Stochastic Games}

\author{Moritz Graf\orcidlink{0009-0001-6561-2046}}
\affiliation{
  \institution{RPTU University Kaiserslautern-Landau \& Max Planck Institute for Software Systems}
  \city{Kaiserslautern}
  \country{Germany}}
\email{moritz.graf@cs.rptu.de}

\author{Anthony Lin\orcidlink{0000-0003-4715-5096}}
\affiliation{
  \institution{RPTU University Kaiserslautern-Landau \& Max Planck Institute for Software Systems}
  \city{Kaiserslautern}
  \country{Germany}}
\email{awlin@mpi-sws.org}

\author{Rupak Majumdar\orcidlink{0000-0003-2136-0542}}
\affiliation{
  \institution{Max Planck Institute for Software Systems}
  \city{Kaiserslautern}
  \country{Germany}}
\email{rupak@mpi-sws.org}

\begin{abstract}
    Many problems in compositional synthesis and verification of multi-agent systems---such as rational verification and assume-guarantee verification in probabilistic systems---reduce to reasoning about two-player multi-objective stochastic
    games.
    This motivates us to study the problem of characterizing the complexity and memory requirements for two-player stochastic games with Boolean combinations of 
    qualitative reachability and safety objectives.
    Reachability objectives require that a given set of states is reached; safety requires that a given set is invariant.
    A qualitative winning condition asks that an objective is satisfied almost surely (AS) or (in negated form) with non-zero (NZ) probability.

    We study the determinacy and complexity landscape of the problem. 
    We show that games with conjunctions of AS and NZ reachability and safety objectives are determined, and determining the winner is PSPACE-complete.
    The same holds for positive boolean combinations of AS reachability and safety, as well as for negations thereof.
    On the other hand, games with full Boolean combinations of qualitative objectives are not determined, and are NEXPTIME-hard.
    Our hardness results show a connection between
    stochastic games and logics with partially-ordered quantification.
    Our results shed light on the relationship between determinacy and complexity, and extend the complexity landscape for stochastic games in the multi-objective setting.
\end{abstract}


\begin{document}
\sloppy


\pagestyle{fancy}
\fancyhead{}


\maketitle 

\section{Introduction}
\begin{table}[t]
\begin{tabular}{c|c|c|c|c}
 $\phi$ & queries & objectives & determined? & complexity \\\hline
 $\bigwedge$ & $AS,NZ$ & $\lozenge,\square$ & yes & PSPACE-complete\\\hline
 $\bigvee$ & $AS,NZ$ & $\lozenge,\square$ & yes & PSPACE-complete \\\hline
 $\mathcal{B}^+$ & $AS$ & $\lozenge,\square$ & yes & PSPACE-complete\\\hline
 $\mathcal{B}^+$ & $NZ$ & $\lozenge,\square$ & yes & PSPACE-complete\\\hline\hline
 $\mathcal{B}^+$ & $AS,NZ$  & $\lozenge,\square$ & no & NEXPTIME-hard\\\hline
 $\mathcal{B}^+$ & \multicolumn{2}{|c|}{$AS\lozenge,\ AS\square,\ NZ\lozenge$} & no & NEXPTIME-complete\\\hline
\end{tabular}
\caption{Main results\label{tab:main}.
Queries that use only conjunctions ($\bigwedge$) or only disjunctions ($\bigvee$) are PSPACE-complete, as are positive Boolean combinations ($\mathcal{B}^+$) of only almost sure ($AS$) or only nonzero ($NZ$) queries. Queries that use conjunctions, disjunctions, and both AS and NZ are NEXPTIME-hard.}
\end{table}

Stochastic games are a class of games consisting of multiple players, in which the environment exhibits stochastic behavior. 
The 2-player version of the game (often called $2\frac{1}{2}$-player games) has especially been extremely useful in modeling many problems in verification, including rational verification in multi-agent systems \cite{kr21,StochasticLossyChannel}, wherein Player 1 represents a deviating player and Player 2 represents a coalition of players who aim to punish the deviating player. However, in many applications of 2-player stochastic games, it is known that \emph{multiple} objectives (i.e. a Boolean combination of them) are necessary (cf. \cite{kr21,OnStochasticGames,DisjunctionSG,prism-games2}).

Many basic problems on solving 2-player stochastic games with multiple objectives are still open. 
Chen et al. \cite{OnStochasticGames} studied this problem, focusing primarily on \emph{quantitative} objectives, i.e., a Boolean combination of expected total reward objectives. 
They proved that such games are not determined. \emph{Determinacy} is the property that for every game and every winning objective, either player 1 wins the winning objective
or player 2 wins the negation of the objective;
it is a generalization of the minimax theorem for two player one-shot games. 
They also show that if players are restricted to deterministic strategies, 
deciding if a player has a winning strategy becomes undecidable. 
For general strategies, PSPACE-hard was shown, but decidability remains open.

Stan et al.\ \cite{StochasticLossyChannel} and Winkler and Weininger
\cite{DisjunctionSG} studied multi-objective stochastic games with
\emph{qualitative} reachability objectives (requiring that a set of states is
reached) and safety objectives (requiring that a given set of states is never
left).
Further, a qualitative condition requires that the underlying reachability or safety objective is 
satisfied almost surely (with probability one; also written $\mathit{AS}$) or with non-zero probability (also written $\mathit{NZ}$).
Thus, winning conditions are general Boolean formulas over the propositions 
\[
(\set{\mathit{AS}, \mathit{NZ}} \times \set{\Diam, \Box}) F,
\]
where $F$ ranges over sets of states, $\Diam$ and $\Box$ denote reachability and safety, respectively, and $\mathit{AS}$ and $\mathit{NZ}$
denote almost surely or non-zero, respectively.

This qualitative setting is, in fact, sufficient for many applications.
For example, this is the case for \emph{liveness} verification for probabilistic distributed protocols \cite{LLMR17,LR16,Lynch-book,Fokkink-book}, e.g., whether a philosopher will eventually eat with probability 1 in a dining philosopher protocol. Such qualitative objectives were also considered in rational verification problems in various settings \cite{StochasticLossyChannel,kr21}. For this qualitative setting, \cite{DisjunctionSG} show that disjunctions of almost sure reachability can be solved in polynomial time, and are between PSPACE and EXPTIME for deterministic strategies. In \cite{StochasticLossyChannel} it is shown that the case with conjunctions of almost sure and nonzero objectives is decidable, and is in EXPTIME. Decidability for the general case with an arbitrary Boolean combination remains an open problem. 

\paragraph{Contributions.} In this paper, we carry out a systematic investigation of the determinacy and complexity of two-player stochastic games with multiple qualitative objectives. 

Our starting point is an observation that different applications give rise to different \emph{classes} of Boolean formulas. For example, in the reduction \cite{kr21,StochasticLossyChannel} from the problem of rational verification for probabilistic systems to two-player stochastic games, one obtains only a \emph{conjunctive} Boolean formula. 
Similarly, in assume-guarantee reasoning, it has been noted in \cite{MMSZ20,OnStochasticGames} that the required Boolean formulas are of the form
\[
    \bigwedge_i \neg \varphi_i \vee \psi_i,
\]
where $\varphi_i$ (resp. $\psi_i$) encodes the required assumption (resp. guaranteed post-condition). 
For these reasons, it makes sense to investigate the problem by \emph{varying the allowed Boolean operators} ($\wedge,\vee,\neg$) in the objectives. Moreover, owing to the duality of $\mathit{AS}(\Diam F)$ and $\mathit{NZ}(\Box F)$
(and similarly the duality of $\mathit{AS}(\Box F)$ and $\mathit{NZ}(\Diam F)$), we may allow only $\wedge,\vee$ (i.e., dispense with negation) and instead vary the individual propositions that are allowed.

For the general case (without restricting the Boolean formulas), we show that such games are not determined. This improves the result in \cite{OnStochasticGames} 
that two-player stochastic games with multiple \emph{quantitative} objectives are not determined. As for the problem of deciding if player 1 has a winning 
strategy, we do not know if this is decidable, but we show a new NEXPTIME lower bound, improving the PSPACE-hardness from \cite{OnStochasticGames}. 
This exploits a connection between dependency quantified Boolean formulas \cite{HenkinDQBF} and stochastic games.

Table \ref{tab:main} summarizes our results for different subclasses of formulas.
For the special case of Boolean combinations of non-zero reachability, we can show decidability and a matching NEXPTIME upper bound.
Our proof uses a characterization of optimal policies: we show that if player 1 has a winning strategy, then player 1 has a winning strategy that uses
exponential memory.

We explore the determinacy boundary.
We show that games with positive Boolean combinations of \emph{only} AS or \emph{only} NZ objectives are determined, as are games where the formula is a pure disjunction or a pure conjunction.
We complement this with a PSPACE upper bound to determine the winner, which matches a PSPACE-hardness inherited from multiple (nonstochastic)
reachability games \cite{fijalkow2010surprizing} or from multi-objective reachability in MDPs \cite{PercentileQueriesMDP}.

\smallskip
\noindent
\textbf{Related Work.}
Very few results were known for multi-objective stochastic games before our work.
This is in contrast to \emph{single-objective} stochastic games---the winning condition is exactly one temporal objective---for which
there is a well-developed algorithmic theory
\cite{Krish-papers-qualitative-and-quantitative,dAHK99,dAH00}.
It is also in contrast to multi-objective problems on Markov decision processes (single-player stochastic games),
where algorithmic solutions are known for Boolean combinations of quantitative LTL objectives \cite{MultiObjectiveMDP}, for percentile queries \cite{PercentileQueriesMDP}, or for combinations probabilistic and non-probabilistic objectives \cite{MixingMDP}.

In addition to the related work on stochastic games that we have discussed earlier, we also mention the work of \cite{ashok2020approximating,DecidabilitySG}. While in general it is only known that the Pareto set for conjunctions of reachability objectives can be approximated \cite{ashok2020approximating}, it is shown in \cite{DecidabilitySG} that total reward is decidable for stopping games with two objectives. Notably, \cite{DecidabilitySG} also give an exponential time algorithm to compute the Pareto set exactly, assuming determinacy.

\paragraph*{Organization.} We define stochastic games in Section \ref{sec:prelim}. 
We provide the proofs of our results for determined (resp. nondetermined) queries in Section \ref{sec:det} (resp. \ref{sec:nondet}). 
Missing proofs are in the \longshort{appendix}{full version}.

\section{Preliminaries}
\label{sec:prelim}
We introduce the setting of \emph{turn-based} stochastic games \cite{Shapley,FilarVrieze,ChatterjeeHenzinger2012}.

\begin{definition}[Stochastic Games]
    A stochastic game is a tuple
    $\mathcal{G} = \langle S, s_0, A, P\rangle$, where
    $S$ is a finite set of states, partitioned into disjoint subsets $S_1$, $S_2$, and $S_c$ controlled by player 1, player 2, and the chance player,
    respectively. Thus, $S = S_1 \cup S_2 \cup S_c$.
    $s_0 \in S$ is an initial state.
    The map $A: S_1 \cup S_2 \rightarrow 2^S\setminus \set{\emptyset}$ maps player 1 and 2 states to possible successors in $S$.
    The map $P: S_c \rightarrow \Delta(S)$ maps states in $S_c$ to a probability distribution over $S$.
\end{definition}
    
    A game starts in the initial state $s_0$ and proceeds in rounds.
    In each round, if the current state is in $S_i$, for $i\in\set{1,2}$ then player $i$ picks a successor $s'$ from $A(s)$ and the new state is $s'$.
    If the current state is in $S_c$, the game moves to a new state $s'\in S$ with probability $P(s)(s')$.
    A state $s$ with $A(s) = \{s\}$ is called a terminal state.
    Note that these games are turn based: in each round, exactly one player decides the next state.

For $s\in S$, we use $\mathcal{G}_s$ to denote the game that is identical to $\mathcal{G}$ but starts at $s$ instead of $s_0$.
We draw a stochastic game $\mathcal{G}$ as a directed graph, with nodes $s\in S$, denoted by a square if $s\in S_1$, a diamond if $s\in S_2$ or a circle if $s\in S_c$, 
and edges $(s, s')$ if $s'\in A(s)$ or $P(s)(s') \neq 0$. 
We will usually omit the self loops on terminal states. 
Figure~\ref{fig:nondetermined} is a stochastic game with $S_1 = \{s_0\}$, $S_2 = \{s_1\}$, $S_c = \{s_2, s_3, s_4\}$.
\begin{figure}[t]
    \centering
    \begin{tikzpicture}[node distance = 1.2cm, on grid, auto]
        \node (s1) [state, rectangle] {$s_0$};
        \node (s2) [state, diamond, above right = of s1] {$s_1$};
        \node (s3) [state, below right = of s1] {$s_2$};
        \node (s4) [state, above right = of s2] {$s_3$};
        \node (s5) [state, below right = of s2] {$s_4$};

        \path[-stealth, thick]
            (s1) edge (s2) 
            (s1) edge (s3)

            (s2) edge (s4)
            (s2) edge (s5);
    \end{tikzpicture}
    \caption{Example stochastic game.}
    \label{fig:nondetermined}
\end{figure}
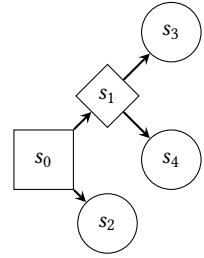

\begin{definition}[Strategies]
    A \emph{strategy} for player 1 in $\mathcal{G}$ is a function
    $\sigma: S^*S_1 \to \Delta(S)$ with
    $\mathit{supp}(\sigma(\pi s)) \subseteq A(s)$, for all $\pi s \in S^*S_1$,
    where $\Delta(S)$ is the set of distributions over $S$ and $\mathit{supp}(d)$ is the support of a distribution $d\in\Delta(S)$.
    Strategies for player 2 are defined symmetrically as functions $\tau: S^*S_2 \to \Delta(S)$.
\end{definition}

When both players fix their strategies, the game induces a Markov chain, where the successors at any state are picked according to the
strategies (or according to $P$ for chance states).
Formally, specifying a pair of strategies $\sigma, \tau$ induces an infinite Markov Chain
$\mathcal{M}_\mathcal{G}^{\sigma, \tau}$, whose nodes are paths in $\mathcal{G}$ starting from $s_0$ following the strategies $\sigma$ and $\tau$ (branching is caused by probabilistic actions). See \cite{chen2012playing} for more details.
The runs of this Markov chain are called \emph{plays} of $\mathcal{G}$ according to $\sigma, \tau$.
We write $\mathcal{P}_{\mathcal{G}}^{\sigma, \tau}$ for the probability measure over these plays, defined by the usual cylinder sets \cite{chen2012playing}.

An \emph{objective} is a set of plays.
We will focus on \emph{reachability} and \emph{safety} objectives.
Let $T\subseteq S$.
The \emph{reachability objective} with target $T$, written $\lozenge T$, is the set of plays in $S^\omega$ that reach $T$.
Its dual,
the \emph{safety objective}, written $\Box T$, is the set of plays that always remain in $T$:
$$\lozenge T = \{\pi \in S^\omega | \exists i: \pi_i\in T\} \quad\mbox{ and }\quad \square T = \set{\pi \in S^\omega \mid \forall i: \pi_i \in T} $$
Note that $\pi \in \lozenge T \iff \pi \not\in \square S\setminus T$.

Fix a game $\mathcal{G}$ and strategies $\sigma$ and $\pi$.
An objective $X\subseteq S^\omega$ is realized \emph{almost surely}, written $\sigma, \tau\models_{\mathcal{G}} \AS X$, 
iff $X$ holds with probability one: $\Pr^{\sigma, \tau}_\mathcal{G}(X) = 1$.
An objective $X$ is realized with \emph{nonzero probability},
written $\sigma, \tau\models_{\mathcal{G}} \NZ X$,
iff $\Pr^{\sigma, \tau}_\mathcal{G}(X) > 0$.

We extend $\AS$ and $\NZ$ queries to Boolean combinations,
with the obvious semantics.
For example, $\sigma, \tau\models_{\mathcal{G}} \AS X \vee \AS Y$
iff $\sigma, \tau\models_{\mathcal{G}} \AS X$
or
$\sigma, \tau\models_{\mathcal{G}} \AS Y$.
We refer to such a Boolean combination as a \emph{query}.

In the following we only consider queries where
the base objectives are either reachability or safety.
Note that for any objective $X$, we have
$\sigma, \tau \models_{\mathcal{G}} \AS X$ iff
$\sigma, \tau\not\models_{\mathcal{G}} \NZ (S^\omega\setminus X)$.
Together with the duality between reachability and safety, 
we see that queries are closed under complementation.

We refer to nontrivial 
Boolean combinations as \emph{multi-objective queries}
and an atomic $\AS$ or $\NZ$ query as a \emph{single-objective query}.

Given a game $\mathcal{G}$ and query $\phi$, 
a strategy $\sigma$ of player 1 is a \emph{winning} strategy
if and only if for every strategy $\tau$ of player 2, 
$\phi$ is satisfied, that is,
$\sigma,\tau\models_{\mathcal{G}} \phi$ holds.
Player 1 is \emph{winning} if they have a winning strategy:
$\exists\sigma\forall\tau: \sigma, \tau \models_{\mathcal{G}} \phi $.
    
Player 2 is \emph{winning} if they have a winning strategy for the negated objective
$\exists\tau\forall\sigma: \sigma, \tau \models_{\mathcal{G}} \lnot\phi $.

For a given query $\phi$, the set $\Sem{\phi}{1} = \{s\in S| \exists\sigma\forall\tau \sigma, \tau \models_{\mathcal{G}_s} \phi \}$ of all states $s$ such that player 1 has a winning strategy if the game starts at $s$ is called the \emph{winning region} of $\phi$. 

A query $\phi$ is \emph{determined} 
if in every game $\mathcal{G}$ with query $\phi$, either player 1 is winning or player 2 is winning.
(We assume the queries and games share the same set of states.)
Determinacy is a non-trivial property of stochastic games.
This is because the logical negation of ``player 1 is winning''
is that $\forall \sigma \exists \tau: \sigma, \tau \models_{\mathcal{G}} \lnot \phi$,
which only ensures that player 2 has a \emph{spoiling strategy}
for any player 1 strategy.
This does not mean that player 2 has a single winning strategy for every player 1 strategy.
The following results hold.

\begin{proposition}[Determinacy]
\begin{enumerate}
\item \cite{dAHK99,dAH00} Games with $\AS X$ and $\NZ X$ objectives for safety and reachability objectives $X$ are determined.
\item {\textsf{[Determinacy Argument]}} A query $\phi$ is determined if 
there exists a determined query $\phi'$ such that for every game $\mathcal{G}$:
        \begin{itemize}
            \item[-] If player 1 wins $\phi'$ then they win $\phi$
            \item[-] If player 2 wins $\phi'$ then they also win $\phi$.
        \end{itemize}
\end{enumerate}
\end{proposition}

The first result follows from a general determinacy theorem for stochastic games for all Borel objectives \cite{Martin98}.
The second result follows because either player 1 or player 2 wins the game with query $\phi'$, and this determines who
wins the game with query $\phi$.

In our proofs we will sometimes restrict the game $\mathcal{G}$ to only a subset of $S$, replacing all other states with a new terminal state $s_\bot$.
\begin{definition}[restricted game]
    Let $\mathcal{G}$ be a stochastic game, and $U\subseteq S$ be a subset of states. We define game $\mathcal{G}| U$ restricted to $U$ as:
    $\langle U\cup \{s_\bot\}, s_0', A', P'\rangle$ where $s_\bot$ is a new terminal state, with:
    \begin{flalign*}
    s_0' & = s_0 \textrm{ if } s_0 \in U,\ s_0' = s_\bot \textrm{ otherwise } & \\
    A'(s) & = A(s) \textrm{ if } A(s) \subseteq U,\ A'(s) = (A(s) \cap U) \cup {s_\bot} \textrm{ otherwise } & \\
    P'(s)(s') & = P(s)(s') \textrm{ if } s' \in U,\ P'(s)(s_\bot) = \sum_{s'\not\in U:} P(s)(s') & 
    \end{flalign*}
    If $X$ is an objective with target set $T$ in $\mathcal{G}$ then we consider it in $\mathcal{G}| U$ with the target set $T' = T\cap U$.
\end{definition}

\section{Determined Queries}
\label{sec:det}

In this section, we prove the following main theorem.

\begin{theorem}
\label{th:determinacy}
Consider the following class of queries:
\begin{enumerate}
\item Conjunctions of $\AS$ and $\NZ$ queries for both
reachability and safety objectives;
\item Positive Boolean combinations of $\AS$ queries for
reachability and safety objectives.
\end{enumerate}
Both classes of queries are determined.
Further, deciding if player 1 wins a game is PSPACE-complete.
\end{theorem}
This result solves the open problem from \cite{kr21} of the precise complexity of rational verification for reachability/safety objectives.

By negating the classes and using determinacy, 
we conclude the following classes are
also determined, and have the same PSPACE-completeness
complexity to determine the winner:
\begin{enumerate}
\item Disjunctions of $\AS$ and $\NZ$ queries for both
reachability and safety objectives;
\item Positive Boolean combinations of $\NZ$ queries for
reachability and safety objectives.
\end{enumerate}

We prove Theorem~\ref{th:determinacy} using a sequence of lemmas,
building up to the main result.
\subsection{Conjunctions}
We start with the first class of conjunctive queries.
\begin{lemma}[Conjunction of nonzero]
\label{lem:conj-nz}
    Let $\phi = \bigwedge_{i\in I} NZ(X_i)$, where $X_i \subset S^\omega$ is any objective.
    $\phi$ is determined and player 1 has a winning strategy for $\phi$ if any only if
    player 1 has a winning strategy for each $\phi_i = NZ(X_i)$.    
\end{lemma}
\begin{proof}
Since each $\phi_i$ is determined, either player 1 has a winning strategy $\sigma_i$ for each $\phi_i$, or there is some $\phi_i$ for which player 2 has a winning strategy $\tau_i$.

If player 1 has a winning strategy $\sigma_i$ for each $i$, then the strategy $\sigma$
that randomizes uniformly between each $\sigma_i$ at the start of the game is a winning strategy for $\phi$, since for all strategies $\tau$ of player 2 we have that
$\Pr^{\sigma, \tau}(X_i) = \frac{1}{|I|} \Pr^{\sigma_i, \tau}(X_i) > 0$.

If player 2 has a winning strategy $\tau_i$ for some $\phi_i$, then this strategy is also a winning strategy for $\phi$, since for any strategy $\sigma$ of player 1 we have that $\sigma, \tau_i\not\models\phi_i$ and therefore $\sigma, \tau_i\not\models\phi$.

It follows that $\phi$ is determined, and player 1 has a winning strategy if and only if they have a winning strategy for each $\phi_i$.
\end{proof}
The following lemma shows that we can reduce a conjunction of multiple almost sure reachability queries to a single almost sure reachability query.
\begin{lemma}[Conjunction of $AS$]
\label{lem:conj-as-reach}
    Let $\phi = \bigwedge_{i=1}^n AS(\lozenge T_i)$ be a conjunction of almost sure reachability queries.
    $\phi$ is determined and player 1 has a winning strategy for $\phi$
    if and only if player 1 has a winning strategy for $\phi' = AS(\lozenge T')$ where
    \begin{align*}
        T' = \bigcup_{i=1}^n T_i \cap \Sem{\bigwedge_{i\neq j} AS(\lozenge T_i)}{1}
    \end{align*}
\end{lemma}
\begin{proof}
    The set $T'$ is a subset of the union of all target states $T = \bigcup_{i=1}^n T_i$ such that if $s\in T'\cap T_i$ is reached, player 1 can guarantee that all other target sets $T_j$ ($j\neq i$) are reached almost surely from each $s$.
    We show by induction over $n$ that the lemma holds for any conjunction with up to $n$ objectives. 
    For the base case, we have since $T' = T$ and $\phi' = \phi = AS(\lozenge T_1)$, so the lemma holds.
    Now assume the lemma holds for any conjunctions of up to $k$ objectives and consider the case of $n = k+1$.
    
    Suppose player 1 has a winning strategy $\sigma$ for $AS(\lozenge T')$. Consider the following strategy: player 1 starts playing $\sigma$ until a state $s\in T'$ is reached. Suppose $s\in T_i \cap \Sem{\bigwedge_{i\neq j} AS(\lozenge T_i)}{1}$ for some $i$ (pick any $i$ in case multiple apply).
    From $s$ player 1 switches to playing the strategy from $\Sem{\bigwedge_{i\neq j} AS(\lozenge T_i)}{1}$.
    
    By the induction hypothesis, each subquery $\bigwedge_{i\neq j} AS(\lozenge T_i)$ is determined. Consider the following strategy: player 2 plays a winning strategy for $\lozenge T'$. If the play reaches some $s\in T_i\setminus T'$ for some $i$, player 2 switches to a winning strategy for $\bigwedge_{i\neq j} AS(\lozenge T_i)$ from s.
    
    It follows that $\phi$ is determined and player 1 has a winning strategy if an only if they have a winning strategy for $AS(\lozenge T')$. Inductively, this holds for conjunctions over any number of objectives. The full proof is in the appendix.
\end{proof}

Next, we extend the results to conjunctions of both
$\AS$ and $\NZ$ reachability.
We start with the special case where there is exactly
one $\NZ$ reachability query.

\begin{lemma}[Conjunction of AS and one NZ reachability]
\label{lem:conj-as-one-nz-reach}
    Let $\phi = NZ(\lozenge T_0) \land \bigwedge_{i=1}^n AS(\lozenge T_i)$
    be a conjunction of one nonzero and multiple almost sure reachability  queries.
    $\phi$ is determined, and deciding if player 1 has a winning strategy is PSPACE-complete.
\end{lemma}

\begin{proof}
    We first construct a nonstochastic reachability game $\langle \mathcal{G}^*, F \rangle$ and show that for each player, winning the query $\phi$ in $\mathcal{G}$ is equivalent to winning the reachability game $\langle \mathcal{G}^*, F \rangle$.
    Since nonstochastic reachability games are determined, it follows that $\phi$ is determined.

    We construct the goal unfolding of $\mathcal{G}$ (see also \cite{DisjunctionSG}, \cite{fijalkow2010surprizing}, \cite{OnStochasticGames}). This is a game constructed from $\mathcal{G}$ with state space $S\times \{0, 1\}^{n+1}$, such that in a state $(s, b)$, the vector $b$ tracks which target sets $T_i$, $i\in \{0, \dots,  n\}$ have already been visited during a play. 
    For each state $s\in S$, let $I_s = \{i\in \{0, \dots n\} | s\in T_i\}$ be the index set of targets that include $s$. We define the goal unfolding formally as $\mathcal{G}' = \langle S', s_0', A', P' \rangle$, where:
    \begin{flalign*}
        S' =\ & S\times \{0, 1\}^{n+1} & \\
        s_0' =\ & (s_0, (0, \dots, 0)) & \\
        A'((s, b)) =\ &\{(s', b') |\ s'\in A(s), &\\
        &\ \forall i\in I_s: b'_i = 1, \forall i\not\in I_s: b'_i = b_i\} & \\
        P'((s, b))((s', b')) =\ & \begin{cases}
            P(s)(s') &\textrm{ if } \forall i\in I_s: b'_i = 1, \forall i\not\in I_s: b'_i = b_i \\
            0 &\textrm{ otherwise}
        \end{cases} &
    \end{flalign*}
    Note that we can identify paths in $\mathcal{G'}$ with paths in $\mathcal{G}$ via projection or by augmenting each state with the vector of previously visited target sets. Similarly we can identify strategies in $\mathcal{G'}$ with strategies in $\mathcal{G}$.

    Consider in $\mathcal{G}'$ the query $\phi' = NZ(\lozenge T_0')\land \bigwedge_{i=1}^n AS(\lozenge T_i')$, where $\forall i\in \{0, \dots, n\}: T_i' = \{(s, b)|s\in S, b_i = 1\}$. This query is equivalent to the query $\phi$ in $\mathcal{G}$, since for any pair of strategies $\sigma, \tau$  it holds that
    $\sigma, \tau \models_{\mathcal{G}'} \phi' $ if and only if
    $\sigma, \tau \models_{\mathcal{G}} \phi $.
    Since all target sets $T_i'$ are absorbing, i.e. any play that reaches $T_i'$ will stay in $T_i'$, it follows that $\sigma, \tau\models \bigwedge_{i=1}^n AS(\lozenge T_i')$ if and only if $\sigma, \tau \models AS(\lozenge T')$ where $T' = \{(s, b)|s\in S, \forall i\neq 0: b_i = 1\}$.

    Let $M = \Sem{AS(\lozenge T')}{1}$ be the winning region of player 1 for $AS(\lozenge T')$ in $\mathcal{G}'$, and let $\mathcal{G}'|\Sem{AS(\lozenge T')}{1}$ be the game restricted to that set.
    Finally, let $\langle \mathcal{G}^*, \lozenge T_0' \rangle$ be the nonstochastic reachability game obtained from the restricted game $\mathcal{G}'|\Sem{AS(\lozenge T')}{1}$ by giving control of the stochastic states to player 1.

    We now show that player 1 has a winning strategy for $\phi'$ in $\mathcal{G'}$ if and only if they win $\langle \mathcal{G}^*, \lozenge T_0' \rangle$. Since $\langle \mathcal{G}^*, \lozenge T_0' \rangle$ is a nonstochastic reachability game it is determined and winning strategies are deterministic and memoryless (in $\mathcal{G}^*$).

    Assume player 2 has a winning strategy $\tau^*$ for $\langle \mathcal{G}^*, \lozenge T_0' \rangle$, and let $\tau$ be the deterministic strategy for player 2 in $\mathcal{G}'$ that plays according to $\tau^*$ until some state $s\not\in M$ is reached, and then switches to a winning strategy $\tau_s$ for player 2 for $AS(\lozenge T')$.    
    Let $\sigma$ be any strategy for player 1 in $\mathcal{G'}$ and assume that $\sigma, \tau \models_{\mathcal{G}'} \phi'$.
    It follows that there is a play $\pi$ with $\Pr^{\sigma, \tau}(\pi) > 0$, $\pi\in \bigcap_{i=0}^n \lozenge T_i'$. If $\pi$ would reach a state $s\not\in M$, then player 2 would play $\tau_s$ and ensure $\sigma, \tau \not\models_{\mathcal{G}'} \bigwedge_{i=1}^n AS(\lozenge T_i')$. Therefore $\pi$ is a play in $M$.    
    Let $\pi^*$ be the corresponding play in $\mathcal{G^*}$ and let $\sigma^*$ be the strategy that on each prefix of $\pi^*$ plays the next state on $\pi^*$ (player 1 plays at all states $(s, b)$ with $s\in S_1'\cup S_c'$). 
    Then it follows that $\pi^* \in \lozenge T_0'$. Since playing $\sigma^*$ against $\tau^*$ in $\mathcal{G^*}$ results in the path $\pi^*$, player 1 wins $\langle \mathcal{G}^*, T_0'\rangle$, which is a contradiction to $\tau^*$ being a winning strategy for $\langle \mathcal{G}^*, \lozenge T_0' \rangle$ for player 2. Therefore it follows that $\sigma, \tau \not\models_{\mathcal{G}} \phi'$, and $\tau$ is a winning strategy for player 2 for $\phi'$ in $\mathcal{G}$.

    Assume now instead that player 1 has a winning strategy $\sigma^*$ for $\langle \mathcal{G}^*, \lozenge T_0' \rangle$ and let $\tau^*$ be any strategy of player 2. Without loss of generality, both $\sigma^*$ and $\tau^*$ are memoryless strategies.    
    Then the play according to $\sigma^*, \tau^*$ is a finite simple path $\pi^*$ that ends in $s\in T_0'$ and therefore does not contain $s_\bot$. The union of all these simple paths form a tree $\Pi$ that branches universally on states of player 2.
    Let $\sigma$ be the strategy for player 1 in $\mathcal{G}$ that plays according to $\sigma^*$ on $\Pi$
    and switches to playing $\sigma_s$ after leaving $\Pi$ at $s$ or after reaching a leaf $s$ of $\Pi$, where $\sigma_s$ is the winning strategy for $\bigwedge_{i=1}^n AS(\lozenge T_i')$ from $s$. Note that a play according to $\sigma$ only leaves $\Pi$ at stochastic states $s\in S_c' \cap M$.
    Let $\tau$ be any strategy of player 2. Because $\Pi$ branches universally at states of player 2 and any play reaching a stochastic state $s \in \Pi$ has a nonzero chance to stay in $\Pi$, there is a root-to-leaf path in $\Pi$ that has positive probability according to $\sigma, \tau$. It follows that $\sigma, \tau \models_{\mathcal{G}'} NZ(\lozenge T_0')$.
    Since any play with positive probability eventually either reaches a leaf of $\Pi \subset M$ or
    leaves $\Pi$ at some stochastic state $s\in S_c'\cap M$, player 1 eventually plays a strategy $\sigma_s$ that is a winning strategy for $\bigwedge_{i=1}^n AS(\lozenge T_i')$. It follows that
    $\sigma, \tau \models_{\mathcal{G}'} \bigwedge_{i=1}^n AS(\lozenge T_i')$. Therefore $\sigma, \tau \models_{\mathcal{G}'} \phi'$ and $\sigma$ is a winning strategy for $\phi'$ in $\mathcal{G}'$.

    Therefore $\phi'$ in $\mathcal{G}'$ is determined, and player 1 has a winning strategy if and only if they win the reachability game $\langle \mathcal{G}^*, \lozenge T_0' \rangle$.

    For complexity, note that the reachability game $\langle \mathcal{G}^*, \lozenge T_0' \rangle$ is constructed from the goal unfolding of $\mathcal{G}'$ and therefore the state space is exponential in terms of the number of almost sure target sets $n$.
    However, because along any path $\pi$, the value of $b_i$ can only change from $0$ to $1$ once, the length of any simple path in $\mathcal{G}^*$ is at most of length $(n+2)|S|$. Therefore the depth of the tree $\Pi$ spanned by any winning strategy in $\mathcal{G}^*$ is only linear in the initial state space $S$, and we can use a polynomial space algorithm to verify if a winning strategy exists by checking if there is a tree in $\mathcal{G}^*$ that branches universally at states of player 2, where each leaf is in $T_0'$. 
    The algorithm can be found in the appendix.
    PSPACE-hardness follows from the multiple almost sure reachability problem in MDPs \cite{PercentileQueriesMDP}    
\end{proof}
    
We now show the general case.

\begin{lemma}[Conjunction of AS and NZ reachability]
\label{lem:conj-as-nz-reach}
    Let $\phi = \bigwedge_{i\in I} AS(\lozenge T_i) \land \bigwedge_{j\in J} NZ(\lozenge T_i) $
    be a conjunction of almost sure and nonzero reachability  queries. 
    $\phi$ is determined, and deciding if player 1 has a winning strategy is PSPACE-complete.
\end{lemma}
\begin{proof}
We show that $\phi$ is determined and that player 1 has a winning strategy for $\phi$ if and only if they have a winning strategy for each $\phi_j = NZ(\lozenge T_j) \land \bigwedge_{i\in I} AS(\lozenge T_i)$, $j\in J$. The rest follows from Lemma~\ref{lem:conj-as-one-nz-reach}.

Assume that for each $j$, player 1 has a winning strategy $\sigma_j$for $\phi_j$. Then the strategy $\sigma$ that plays each $\sigma_j$ with probability $\frac{1}{|J|}$ is a winning strategy for $\phi$:
Let $\tau$ be any strategy of player 2, then $\sigma_j, \tau \models \bigwedge_{i\in I} AS(\lozenge T_i)$ for each $j\in J$, therefore $\sigma, \tau \models \bigwedge_{i\in I} AS(\lozenge T_i)$, and $\sigma_j, \tau \models  NZ(\lozenge T_j)$ therefore $\sigma, \tau \models  NZ(\lozenge T_j)$ for all $j\in J$.

Assume now instead that for some $j$, player 2 has a winning strategy $\tau_j$ for $\phi_j$. For any strategy $\sigma$ of player 1, it follows that $\sigma, \tau_j \not\models \phi_j$ and therefore $\sigma, \tau_j\not\models \phi$. Therefore $\tau_j$ is a winning strategy for player 2 for $\phi$. 

It follows that $\phi$ is determined and player 1 has a winning strategy if and only if they have a winning strategy for each $\phi_j$.
\end{proof}
    Note that the approach from Lemma~\ref{lem:conj-as-nz-reach} that generalizes from a single nonzero reachability objective to conjunctions involving multiple nonzero reachability objectives also works for nonzero safety objectives.

Finally, we add safety conditions.
For almost sure safety, first note that for any pair of strategies $\sigma, \tau$, it holds that 
$\sigma, \tau \models \bigwedge AS(\square T_i)$ iff $\sigma, \tau \models AS(\bigcap \square T_i)$
and $\bigcap \square T_i = \square \bigcap T_i$. Therefore conjunctions of multiple almost sure safety queries can be replaced with a single almost sure safety query.
\begin{lemma}[Almost sure safety]
    \label{lem:as-safe}
    Let $\phi = \psi \land AS(\square T)$, where $\psi$ is positive Boolean formula over almost sure and nonzero reachability queries. Let $\mathcal{G'} = \mathcal{G}| \Sem{AS(\square T)}{1}$
    be the game restricted to the winning region of $AS(\square T)$.
    Then $\phi$ is determined, and player 1 has a winning strategy if and only if player 1 has a winning strategy for $\psi$ in $\mathcal{G'}$.
\end{lemma}
We give a proof in the appendix. The proof only shows that Lemma~\ref{lem:as-safe} holds for conjunctions with reachability targets. However, since almost sure safety in stochastic games is equivalent to safety in the non-stochastic game where player 2 is given control of the stochastic states, Lemma~\ref{lem:as-safe} actually holds for all $\psi$.

For nonzero safety, we use a result from \cite{StochasticLossyChannel}, which shows that in conjunctions, we can replace a nonzero safety objective with a nonzero reachability objective. Note that in general this result only applies to nonzero safety objectives where a play that leaves the target set $T$ can not enter $T$ again later. The construction of the goal unfolding ensures that this holds.
\begin{lemma}[Conjunction of AS and one NZ safety] 
    \label{lem:nz-safe}   
    Let $\phi = NZ(\square T_0) \land \bigwedge_{i=1}^n AS(\lozenge T_i)$
    be a conjunction of one nonzero safety query and multiple almost sure reachability queries.
    $\phi$ is determined, and deciding if player 1 has a winning strategy is PSPACE-complete.
\end{lemma}
\begin{proof}
    Similar to the proof of Lemma~\ref{lem:conj-as-one-nz-reach}, we can construct the goal unfolding $\mathcal{G}'$, and consider the equivalent query
    \begin{align*}
        \phi' = NZ(\square T_0') \land \bigwedge_{i=1}^n AS(\lozenge T_i')
    \end{align*}
    with $T_i' = \{(s, b)|s\in S, b_i = 1\}$ for $i\in \{1, \dots, n\}$ as in Lemma~\ref{lem:conj-as-one-nz-reach} and
    $T_0' = \{(s, b) | s\in S, b_0 = 0 \}$.
    Here we can apply a result from \cite{StochasticLossyChannel} that states that player 1 wins $\phi'$ if and only if they win $\phi'' = NZ(\lozenge T_0'') \bigwedge_{i=1}^n AS(\lozenge T_i')$
    where
    \begin{align*}
        T_0'' = [\![NZ(\square T_0')]\!]_1 \cap \bigcap_{i=1}^n T_i'
    \end{align*}
    A winning strategy $\sigma$ for $\phi''$ reaches a state
    $(s, b) \in [\![NZ(\square T_0')]\!]_1\cap \bigcap_{i=1}^n T_i$ with positive probability. Since $b_0 = 0$, it follows that for each state $(s', b')$ in the play so far $b'_0 = 0$ and therefore the current play is in $\square T_0'$. Since the almost sure reachability queries are already satisfied, player 1 can then switch to playing the winning strategy from
    $\Sem{NZ(\square T_0')}{1}$ to satisfy $NZ(\square T_0')$.

    The rest of the lemma follows from Lemma~\ref{lem:conj-as-one-nz-reach}, since now $\phi''$ is a conjunction of one nonzero reachability and multiple almost sure reachability objectives.
\end{proof}
Finally, the proof for Theorem~\ref{th:determinacy} for a conjunction $\phi$ of almost sure and nonzero reachability and safety queries is as follows:
\begin{enumerate}
    \item For each nonzero query $NZ(X_i)$ in $\phi$, consider the conjunction of $NZ(X_i)$ and all almost sure queries in the goal unfolding.
    \item Use Lemma~\ref{lem:as-safe} to remove almost sure safety queries by moving to a restricted game.
    \item If $X_i = \square T_i$, use Lemma~\ref{lem:nz-safe} to replace any nonzero safety queries with nonzero reachability queries.    
    \item Use Lemma~\ref{lem:conj-as-one-nz-reach} to check if player 1 has a winning strategy for the resulting conjunction of one nonzero and multiple almost sure reachability queries.
    \item With Lemma~\ref{lem:conj-as-reach}, player 1 has a winning strategy for $\phi$ if and only if they have a winning strategy for each of these conjunctions.
\end{enumerate}

\subsection{Positive Boolean combinations}
We now move on to positive Boolean combinations.
First note that we can negate Lemma~\ref{lem:conj-nz} to get
the following corollary for 
disjunctions over multiple almost sure queries.

\begin{corollary}[Disjunction Almost-sure]
\label{cor:disj-as}
    Let $\phi$ be a disjunction of almost sure queries. Then
    $\phi$ is determined and player 1 has a winning strategy for $\phi=\bigvee_{i=1}^n \AS(X_i)$ if and only if they have a winning strategy for some $\phi_i = AS(X_i)$.
\end{corollary}

Determinacy follows because the negation of a determined query is still determined. 
A separate proof for this can also be found in \cite{DisjunctionSG}.
Together with Lemma~\ref{lem:conj-as-reach}, it follows that positive 
Boolean combinations of almost sure reachability and safety queries are determined and PSPACE-complete.

\begin{lemma}[Positive Boolean $\AS$-reachability and $\AS$-safety]
    Let $\phi$ be a positive Boolean formula over almost-sure reachability and safety queries.
    $\phi$ is determined, and deciding which player has a winning strategy is PSPACE-complete.
\end{lemma}

\begin{proof}
For determinancy convert $\phi$ to DNF and note that for any conjunction of almost sure objectives it holds that
$\sigma, \tau \models AS(X_1) \land AS(X_2)$ if and only if $\sigma, \tau \models AS(X_1 \land X_2)$.
Therefore $\phi$ is equivalent to a disjunction of almost sure queries.
Determinacy follows from Corollary~\ref{cor:disj-as}.

For membership in PSPACE, consider the following algorithm.
First, nondeterministically guess a satisfying assignment to $\phi$ and Let
$\phi' = \bigwedge_{i\in I} AS(\lozenge T_i) \land \bigwedge_{j\in J} AS(\square T_j)$
be a conjunction over the positive variables in that assignment.
From Lemma~\ref{lem:as-safe}, we know that player 1 has a winning strategy if and only if they have a winning strategy for
$\phi'' = \bigwedge_{i\in I} AS(\lozenge T_i)$ in the game $\mathcal{G}|\Sem{AS(\square \bigcap_{j\in J} T_j)}{1}$.
Since $\phi''$ is a conjunction of almost sure reachability queries, we can decide if player 1 has a winning strategy in PSPACE using Lemma~\ref{lem:conj-as-reach}.
If player 1 has no winning strategy for any satisfying assignment of $\phi$,
then it follows from Corollary~\ref{cor:disj-as} that player 1 has no winning strategy for $\phi$.
\end{proof}
\section{Nondetermined Queries}
\label{sec:nondet}

In this section, we show that the subclasses shown to be determined in Theorem~\ref{th:determinacy}, and their negations, form a maximal class.
That is, games with queries outside these classes are not determined: neither player may have a winning strategy.
Moreover, we show that the decision problem of determining if player 1 can win is at least NEXPTIME-hard.

\OMIT{
The previous section has shown that queries that only use almost sure queries (or only nonzero queries), or only use conjunctions (or only disjunctions) are determined and PSPACE complete.
In this section we show that queries that include both almost sure and nonzero queries, as well as both conjunctions and disjunctions can be nondetermined.
We show that this loss of determinancy corresponds to an increase in complexity be reducing the problem of satisfiability in Dependency Quantified Boolean Formulas (DQBF)\cite{HenkinDQBF} to deciding if player 1 has a winning strategy for a with a positive Boolean combination of almost sure and nonzero reachability queries in a stochastic game.
}

\OMIT{
For positive Boolean combinations of almost sure reachability, almost sure safety and nonzero reachability (excluding only nonzero safety), we show that winning strategies require at most an exponential number of memory states. It follows that in this case, deciding if player 1 has a winning strategy is $NEXPTIME-complete$.
The decidability of the full problem remains open.
}

We start with an example of a nondetermined query.
Consider the with the game from Figure~\ref{fig:nondetermined}. 
We show that the query $\phi = AS(\lozenge \{s_3\}) \lor (NZ(\lozenge \{s_2\})\land NZ(\lozenge \{s_4\}))$ is not determined.
The available strategies for both players in this game can be characterized as follows: player 1 either plays a strategy $\sigma_1$ that plays $s_1$ with probability 1, or a strategy $\sigma_2$ that plays $s_2$ with probability greater than 0. Symmetrically player 2 can either play a strategy $\tau_1$ than plays $s_3$ with probability 1, or a strategy $\tau_2$ that plays $s_4$ with positive probability.
It follows that $\sigma_1, \tau_1 \models \phi$ and $\sigma_2, \tau_2 \models \phi$, but $\sigma_1, \tau_2 \not\models \phi$ and $\sigma_2, \tau_1 \not\models \phi$. Therefore there are no winning strategies for either player, and $\phi$ is not determined.

This nondeterminancy example extends to the following classes of queries:
\begin{enumerate}
    \item positive Boolean combinations of AS and NZ reachability (example);
    \item Boolean combination of AS reachability with
        \begin{align*}
            \phi' = AS(\lozenge \{s_3\}) \lor (\lnot AS(\lozenge \{s_3, s_4\})\land NZ(\lozenge \{s_4\}))
        \end{align*}
    \item Boolean combinations of NZ reachability 
        \begin{align*}
           \phi'' = \lnot NZ(\lozenge \{s_2, s_4\}) \lor (NZ(\lozenge \{s_2\})\land NZ(\lozenge \{s_4\}))
        \end{align*}
\end{enumerate}
as well as for the qualitative multiple safety queries that arise from their respective negations.

\subsection{Hardness}

While the determined queries in the previous sections were PSPACE-complete, we show the extended classes of queries are NEXPTIME-hard.
To show  NEXPTIME-hardness, we use a reduction from Boolean formulas with Henkin quantifiers, 
known as \emph{Dependency Quantified Boolean Formulas (DQBF)} \cite{HenkinDQBF}. 
In comparison to a regular QBF, where a quantified variable always depends on exactly the previously quantified variables, 
in DQBF this dependency can be specified explicitly. We are concerned with S-form DQBFs:
$$ \Phi = \forall x_1, \dots , \forall x_n \exists y_{1, S_1} , \dots, \exists y_{m, S_m} \phi $$
where $X = \{x_1, \dots, x_n\}$ are the universally quantified variables, $Y = \{y_1, \dots, y_m\}$ the existentially quantified variables, and for each $j\in \{1, \dots, m\}$, the set $S_j\subset X$ is the set of variables that $y_j$ depends on. 
This formula is satisfied if and only if there are Skolem functions $Y_j: S_j \to \{0, 1\}$, such that the Skolemization $\forall x_1, \dots x_n: \phi[y_j \to Y_j(X)]$ is satisfied, where
$\phi[y_j \to Y_j(X)]$ denotes $\phi$ with each $y_j$ replaced by $Y_j(X)$. 
Satisfiability of a DQBF in S-Form is NEXPTIME-complete \cite{HenkinDQBF}.

\begin{theorem}[reduction from DQBF]
    \label{lem:hardness}
    Let $\phi$ be a positive Boolean combination of AS and NZ reachability queries. Deciding if player 1 has a winning strategy is NEXPTIME-hard.
\end{theorem}
We describe the construction informally here, the full proof can be found in the appendix.
For a DQBF $\Phi$ in S-Form:
$$ \Phi = \forall x_1, \dots , \forall x_n \exists y_{1, S_1} , \dots, \exists y_{m, S_m} \phi$$
We construct a stochastic game where the initial state $s_0$ randomizes between $m$ branches.

Each branch $j\in\{1, \dots, m\}$ consists of $n+m$ modules, one for each variable. In each module for a variable $v\in X\cup Y$ we have a state controlled by player 1 if $v\in Y$ and by player 2 if $v\in X$, where the player chooses between a state that corresponds to setting the variable $v$ to true, and a state corresponding to setting the variable $v$ to false.

The modules in each branch are ordered according to the dependency induced by $S_j$: 
Each branch $j$ begins with the modules for $x\in S_j$, followed by the module for $y_j$, 
then the remaining modules $x\not\in S_j$ and finally the modules $y_i,\ i\neq j$.
The winning condition for the game $\psi$ is constructed by augmenting the formula $\phi$ with restrictions $\psi_1$ and $\psi_2$ that ensure that each player plays deterministically and identically on each branch:
\begin{align*}
    \psi = (\phi' \land \psi_1) \lor \psi_2
\end{align*}
where $\phi'$ is obtained from $\phi$ by replacing each positive variable $v$ with the almost sure reachability query for the states that set $v$ to true, and $\lnot v$ with the almost sure reachability query for the states that set $v$ to false.

Player 1 wins if player 2 does not follow the restriction (making $\psi_2$ true) or if they follow the restriction (making $\psi_1$ true) and the variable assignment determined by $\sigma, \tau$ satisfies $\phi'$. To show that player 1 has a winning strategy for $\psi$ in $\mathcal{G}_\Phi$, we use that a strategy of player 2 that follows these restrictions can be translated to an assignment to the variables $X$, and a strategy of player 1 strategy that satisfies these restrictions can be translated to Skolem functions $Y_j: S_j\to \{0, 1\}$.

Essentially the same proof works for any of the other query types previously shown to be nondetermined 
(e.g. Boolean combinations of AS reachability, Boolean combinations of NZ reachability, etc.).

\subsection{Membership}

While in general a strategy is a function depending on the entire play of the game so far, 
often strategies only need to remember limited information about the history of the play.
A strategy can be realized by a \emph{strategy automaton} \cite{DisjunctionSG} with a state space $M$ called memory. 
In each round, the strategy automaton updates its memory state $m\in M$ (potentially in a probabilistic way) based on the new state $s$ and, 
if it is the respective player's turn, outputs a distribution $\Delta(S)$ based on its memory state $m$ and the game state $s$.
For a strategy $\sigma$, the smallest $k\in \mathbb{N} \cup \{\infty\}$ such that there is a strategy automaton 
realizing $\sigma$ with $|M| = k$ is the memory size of $\sigma$.

To put an upper bound on the complexity, we aim at bounding the memory requirement:
Given a bound on the memory of a winning strategy, we can nondeterministically
guess a strategy of that bound and check if the strategy is winning in the induced MDP using the algorithms of \cite{MultiObjectiveMDP}.

We will show that for positive Boolean combinations that include $AS(\lozenge)$, $AS(\square)$ and $NZ(\lozenge)$ queries only require exponential memory, and therefore deciding if player 1 has a winning strategy is in NEXPTIME. 
For queries that also include $NZ(\square)$ the memory requirement, and hence decidability/complexity remains open.

In generalized reachability games \cite{fijalkow2010surprizing},
it is known that winning strategies only need to remember the set of target sets previously visited. 
We show that this is not sufficient for our class of queries.
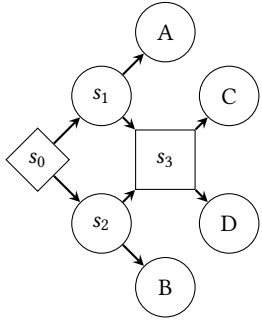
\begin{figure}[t]
    \centering
     \begin{tikzpicture}[node distance = 1.2cm, on grid, auto]
        \node (s1) [state, diamond] {$s_0$};
        \node (s2) [state, above right = of s1] {$s_1$};
        \node (s3) [state, below right = of s1] {$s_2$};
        \node (s4) [state, above right = of s2] {A};
        \node (s5) [state, below right = of s3] {B};
        \node (s6) [state, above right = of s3, rectangle] {$s_3$};
        \node (s7) [state, above right = of s6] {C};
        \node (s8) [state, below right = of s6] {D};

        \path[-stealth, thick]
            (s1) edge (s2) 
            (s1) edge (s3)

            (s2) edge (s4) 
            (s3) edge (s5)

            (s2) edge (s6) 
            (s3) edge (s6)

            (s6) edge (s7)
            (s6) edge (s8);
    \end{tikzpicture}
    \caption{Example game where a winning strategy requires memory other than visited target sets.}
    \label{fig:targetmemory}
\end{figure}

Consider in the game from Figure~(\ref{fig:targetmemory}) the objective
\begin{flalign*}
    (NZ(\lozenge A) \land NZ(\lozenge B))
    &\lor  (AS(\lozenge \{A, B, D\}) \land NZ(\lozenge D)) \\
    &\lor (AS(\lozenge \{A, B, C\}) \land NZ(\lozenge C)) 
\end{flalign*}
Since all target sets are terminal states, if memory of the target sets visited was sufficient, then player 1 should have a winning strategy if and only if they have a memoryless winning strategy.
If player 2 plays both $s_1$ and $s_2$ with positive probability, then $(NZ(\lozenge A) \land NZ(\lozenge B))$ and therefore $\phi$ is satisfied, regardless of player 1s actions.
Therefore we only have to consider spoiling strategies $\tau_A$ and$\tau_B$ that play $s_1$ and $s_2$ with probability 1 respectively.
For player 1, any strategy that randomizes between $C$ and $D$ will fail to satisfy $\phi$ against both $\tau_B$ and $\tau_A$.
Therefore there are two possible memoryless strategies for player 1: $\sigma_C$ and $\sigma_D$ which play $C$ and $D$ with probability 1 respectively, and neither is a winning strategy:
$\sigma_C, \tau_B \not\models \phi$ and $\sigma_D, \tau_A \not\models \phi$.
Player 1 does however have a winning strategy that uses memory: If $\sigma$ plays $C$ with probability 1 if $s_1$ has been played,
and plays $D$ with probability 1 if $s_2$ has been played. This strategy wins against $\tau_A$ by playing $C$, wins against $\tau_B$ by playing $D$ and trivially wins against any randomization.

For queries that do not include $NZ(\square)$, we show that remembering the subset of all visited states (not just target sets) is sufficient for a winning strategy. 
To show this, we assume there is a winning strategy $\sigma$ (with arbitrary memory) and construct a strategy $\bsigma$ from $\sigma$ that only uses the set of previously visited states as memory, and show that $\bsigma$ is still a winning strategy.

Let $\sigma$ be a strategy.
For a play $\pi$, let $\mathit{set}(\pi) =\{s\in S\mid \exists i: \pi_i = s\}$ be the set of states visited by $\pi$.
For a state $s\in S$ and set of states $M\subset S$, define the set $\mathit{resp}_\sigma(s, M)$ as:
\begin{flalign*}
\{s'\in A(s) \mid \exists \tau',\pi':
    \Pr^{\sigma, \tau'}(\pi's) > 0 
    \ \mathit{set}(\pi') = M 
    \sigma(\pi's)(s') > 0 \} 
\end{flalign*}
i.e., the actions that are played with positive probability by $\sigma$ after some play that reaches $s$ and visits exactly the states in $M$.
Let $\bsigma$ be the strategy derived from $\sigma$ where:
\begin{flalign*}
    \bsigma(\pi s)(s') = \begin{cases}
        \frac{1}{\mid \mathit{resp}_\sigma(s, \mathit{set}(\pi))\mid} 
            &\textrm{ if } s'\in \mathit{resp}_\sigma(s, \mathit{set}(\pi)) \\        
        0   &\textrm{ otherwise}
    \end{cases}
\end{flalign*}

\begin{lemma}
The strategy $\bsigma$ derived from $\sigma$ is well defined, and only uses exponential memory.
\end{lemma}
\begin{proof}
    $\bsigma(\pi s)$ only depends on $\mathit{resp}_\sigma(s,
    \mathit{set}(\pi))$ and so only on $s$ and $\mathit{set}(\pi)$. Thus, the only memory $\bsigma$ requires is the set of previously visited states, which is exponential in the number of states.

    Let $\tau$ be any strategy for player 2 and $\pi s$ be a path such that $\Pr^{\bsigma, \tau}(\pi s) > 0$.
    We show that the set $\mathit{resp}_\sigma(s, \mathit{set}(\pi))$ is
    non-empty and thus $\bsigma$ is well defined. For this we prove:
    \begin{align*}
        \exists \tau', \pi': & Pr^{\sigma, \tau'}(\pi's) > 0,\ \mathit{set}(\pi) = \mathit{set}(\pi')\ \ \ \ \ \ (*)
    \end{align*}
    by induction over finite paths $\pi s$.
    
    $(*)$ trivially holds for $\pi s = \langle s_0 \rangle$ with $\tau' = \tau$ and $\pi' = \pi$.
    
    Assume $(*)$ holds for some $\pi s$ with strategy $\tau'$ and alternative path $\pi'$, and consider the path $\pi s s'$ for $s'\in A(s)$:
    
    If $s'\in S_c$, then by assumption $\Pr^{\sigma, \tau'}(\pi's) > 0$ and thus $\Pr^{\sigma, \tau'}(\pi'ss') = \Pr^{\sigma, \tau'}(\pi's) P(s)(s') > 0$.
    Since $\mathit{set}(\pi) = \mathit{set}(\pi')$ it also follows that $
    \mathit{set}(\pi s) = \mathit{set}(\pi' s')$ and so $(*)$ holds for $\pi s s'$.
    
    If $s'\in S_2$ and $\Pr^{\bsigma, \tau}(\pi ss') > 0$, then consider the strategy $\tau''$ that plays like $\tau'$ on the path $\pi'$ and then plays like $\tau$ everywhere else.
    It follows that $\Pr^{\sigma, \tau''}(\pi'ss') = \Pr^{\sigma, \tau'}(\pi's) \tau(\pi s) (s')> 0$.
    Since $\mathit{set}(\pi) = \mathit{set}(\pi')$  by assumption, it also follows that $ \mathit{set}(\pi s) = \mathit{set}(\pi' s)$ 
    and therefore $(*)$ holds for $\pi s s'$.
    
    If $s'\in S_1$ and $\Pr^{\bsigma, \tau}(\pi ss') > 0$, then it follows that $s'\in \mathit{resp}_\sigma(s, \mathit{set}(\pi))$. 
    By construction there exist $\tau', \pi'$ such that
    $\Pr^{\sigma, \tau'}(\pi's) > 0,\ \mathit{set}(\pi) = \mathit{set}(\pi')$, and $\sigma(\pi's)(s) > 0$. It follows that $\Pr^{\sigma, \tau'}(\pi ss') > 0$ and $\mathit{set}(\pi s) = \mathit{set}(\pi' s)$, so $(*)$ holds for $\pi ss'$.
    
    It follows that for any path $\pi s$ with $\Pr^{\bsigma, \tau}(\pi s) > 0$, the set $\mathit{resp}_\sigma(s, \mathit{set}(\pi))$ is non-empty and therefore $\bsigma$ is well defined.
\end{proof}

\begin{lemma}[exponential memory]
    \label{lem:bsigma_win}
    Let $\phi$ be a positive Boolean combination of AS reachability, AS safety and NZ reachability queries.
    If $\sigma$ is a winning strategy for $\phi$, then $\bsigma$ is a winning strategy for $\phi$.
\end{lemma}
\begin{proof}
    We prove that $\bsigma$ is a winning strategy by assuming that there is a strategy $\tau$ such that $\bsigma, \tau \not\models \phi$ and
    showing a contradiction by constructing $\tau^*$ with $\sigma, \tau^* \not\models \phi$.
    
    This proof uses two important properties of $\bsigma$.
    First, from the previous proposition we know if there is a play $\pi s$ according to $\bsigma, \tau$,
    then there is a play $\pi's$ according to $\sigma, \tau'$ that visits the same set of states.
    Additionally from the construction of $\bsigma$ it is clear that any finite play
    that has positive probability according to $\sigma, \tau$ also has positive probability according to $\bsigma, \tau$.
    
    From the second property it follows that for any target set $T$ we have $\sigma, \tau \models NZ(\lozenge T) \implies \bsigma, \tau \models NZ(\lozenge T)$.
    Therefore there must be some almost sure safety or reachability queries $\phi_i$, $i\in I^*$ in $\phi$ such that $\bsigma, \tau \not\models \phi_i$.
    Using the first property, we then construct a strategy $\tau^*_i$ for each of these queries such that $\sigma, tau^*_i\not\models \phi_i$,
    and let $\tau^*$ be the strategy that chooses to play any $\tau^*_i$ at random.
    
    We then show that $\sigma, \tau^* \models NZ(\lozenge T) \implies \bsigma, \tau \models NZ(\lozenge T)$.    
    It follows that $\sigma, \tau^* \not\models \phi$ which is a contradiction to $\sigma$ being a winning strategy.
    Therefore the assumption of $\bsigma$ not winning against $\tau$ is false, and $\bsigma$ is a winning strategy.
    Details are in the appendix.
\end{proof}

This construction does however not work with nonzero safety targets,
as $\sigma, \tau \models NZ(\square T) \implies \bsigma, \tau \models NZ(\square T)$ does not hold.
This can be seen in a simple example game where player 1 can choose to either stay in the initial state $s_0$ or choose to move to a terminal state $s_1$.
A strategy $\sigma_k$, $k\in \mathbb{N}$ that randomizes between staying in $s_0$ and moving to $s_1$ for the first $k$ iterations,
and then chooses to play $s_0$ forever is a winning strategy for $NZ(\square s_0)$,
but $\bsigma$ constructed from $\sigma$ would randomize uniformly between $s_0$ and $s_1$ at every step,
eventually reaching $s_1$ with probability 1 and not satisfying $NZ(\square s_0)$.

Since exponential memory is sufficient for positive Boolean combinations that do not include $NZ(\square T)$ queries,
we can decide if player 1 has a winning strategy in NEXPTIME.

\begin{theorem}
    Let $\phi$ be a positive Boolean combination of almost sure reachability, almost sure safety and nonzero reachability queries.
    Deciding if player 1 has a winning strategy is NEXPTIME-complete.
\end{theorem}
\begin{proof}
    From the previous proposition we know that if a winning strategy $\sigma$ exists,
    then there exists a winning strategy $\bsigma$ that only uses exponential memory, and randomizes uniformly at every state.
    We guess such an exponentially sized strategy $\bsigma$,
    and verify in exponential time that player 2 does not have a winning strategy for $\lnot\phi$ in the MDP $\mathcal{G}_{\bsigma}$,
    which can be done in exponential time using the algorithm from \cite{MultiObjectiveMDP}.
\end{proof}

It remains open whether positive Boolean combinations including $NZ\square$ can be decided in NEXPTIME,
but we motivate why it is not trivial to extend our result.
Consider again the previous example where $\bsigma$ fails for a simple $NZ(\square T)$ in a game with 2 states.
The two strategies $\sigma_1(s_0^k)(s_1) = \frac{1}{k}$
and $\sigma_1(s_0^k)(s_1) = \frac{1}{k^2}$ both play $s_1$ with positive probability at every step, but with arbitrarily small probability overall.
However, the probability of reaching $s_1$ with $\sigma_1$ is $\lim \frac{k - 1}{k} = 1$,
while the probability of reaching $s_1$ with $\sigma_2$ is $\lim \frac{k - 1}{2k} = \frac{1}{2} < 1$.
A construction of $\bsigma$ that works for queries that include $NZ\square$ would need to be able to distinguish between these two strategies.

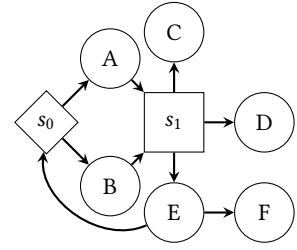
\begin{figure}[t]
    \centering
    \begin{tikzpicture}[scale=0.4, node distance = 1.2cm, on grid, auto]
        \node (s1) [state, diamond] {$s_0$};
        \node (s2) [state, above right = of s1] {A};
        \node (s3) [state, below right = of s1] {B};
        \node (s4) [state, above right = of s3, rectangle] {$s_1$};

        \node (s5) [state, above = of s4] {C};
        \node (s6) [state, right = of s4] {D};

        \node (s7) [state, below = of s4] {E};
        \node (s8) [state, right = of s7] {F};

        \path[-stealth, thick]
            (s1) edge (s2) 
            (s1) edge (s3)

            (s2) edge (s4) 
            (s3) edge (s4)

            (s4) edge (s5)
            (s4) edge (s6)
            (s4) edge (s7)

            (s7) edge (s8)
            (s7) edge[bend left=60] (s1);
    \end{tikzpicture}
    \caption{Example game where winning strategies require memory other than the set of visited states.}
    \label{fig:statememory}
\end{figure}

Additionally, we show that simply remembering all visited states is not sufficient if we include $\NZ(\square)$ queries.
Consider in the game from Figure~(\ref{fig:statememory}) the objective $\phi = \phi_1 \lor \phi_2 \lor \phi_3 \lor \phi_4$ where:
\begin{flalign*}
    \phi_1 = & AS(\lozenge A) \land NZ(\lozenge B) \land NZ(\lozenge C) \land AS(\square \overline{D}) & \\
    \phi_1 = & NZ(\lozenge A) \land AS(\lozenge B) \land AS(\square \overline{C}) \land NZ(\lozenge D) & \\
    \phi_3 = & NZ(\square \overline{A}) \land NZ(\square \overline{B}) & \\
    \phi_4 = & AS(\lozenge F) \land (AS(\square \overline{A})\lor AS(\square \overline{B})) &
\end{flalign*}
$\phi_3$ ensures that at the start of the game, player 2 plays either $A$ or $B$ with probability 1, while $\phi_4$ ensures that if player 1 never plays $C$ or $D$
(and therefore the game almost surely reaches $F$), player 2 has to eventually both $A$ and $B$.
The strategy $\sigma$ that plays $E$ with probability 1 until both $A$ and $B$ have been reached and then plays $C$ if $A$ was reached before $B$ and $D$ otherwise is a winning strategy.
However, a strategy that only uses the set of visited states as memory can not distinguish between these two cases, 
and there is no winning strategy with that memory structure.
We do not know if a different exponential-sized memory structure is sufficient to solve queries that include $NZ(\square)$ and therefore if these queries are also in NEXPTIME.



\begin{acks}
    We thank Daniel Stan and anonymous reviewers for their valuable feedback.
    This work is in part supported by the EU 
    \includegraphics[width=0.50cm]{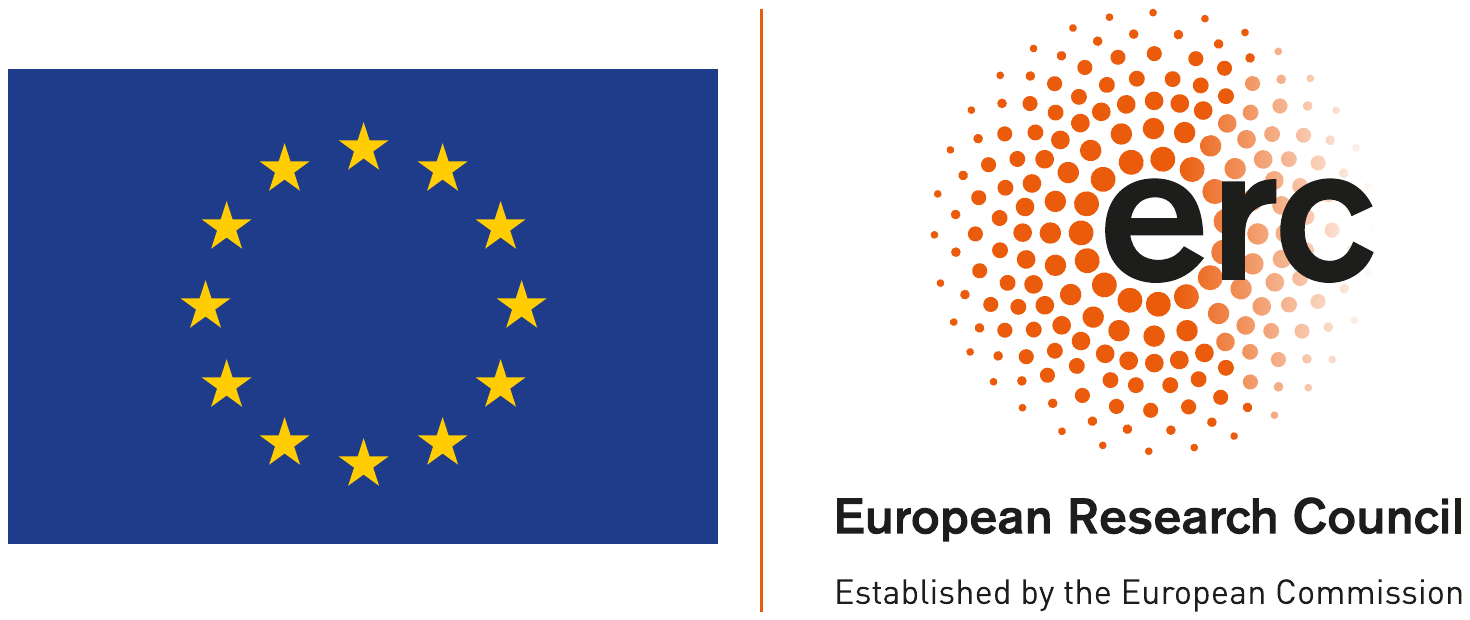}
    (ERC, LASD, grant number 101089343,
    \url{https://doi.org/10.3030/101089343}) and the 
    DFG project 389792660 TRR 248–CPEC.
    Views and opinions expressed are however those
    of the author(s) only and do not necessarily reflect those of the European
    Union or the European Research Council Executive Agency. Neither the
    European Union nor the granting authority can be held responsible for them.
\end{acks}


\balance
\bibliographystyle{ACM-Reference-Format} 
\bibliography{refs.bib}

@inproceedings{Krish-papers-qualitative-and-quantitative,
  author       = {Krishnendu Chatterjee and
                  Marcin Jurdzinski and
                  Thomas A. Henzinger},
  editor       = {J. Ian Munro},
  title        = {Quantitative stochastic parity games},
  booktitle    = {Proceedings of the Fifteenth Annual {ACM-SIAM} Symposium on Discrete
                  Algorithms, {SODA} 2004, New Orleans, Louisiana, USA, January 11-14,
                  2004},
  pages        = {121--130},
  publisher    = {{SIAM}},
  year         = {2004},
  url          = {http://dl.acm.org/citation.cfm?id=982792.982808},
  bibsource    = {dblp computer science bibliography, https://dblp.org}
}

@book{FilarVrieze,
title = "Competitive Markov Decision Processes - Theory, Algorithms and Applications",
author = "Jerzy Filar and Koos Vrieze",
year = "1997",
language = "English",
publisher = "Springer ",

}

@article{ChatterjeeHenzinger2012,
  author       = {Krishnendu Chatterjee and
                  Thomas A. Henzinger},
  title        = {A survey of stochastic {\(\omega\)}-regular games},
  journal      = {J. Comput. Syst. Sci.},
  volume       = {78},
  number       = {2},
  pages        = {394--413},
  year         = {2012},
  url          = {https://doi.org/10.1016/j.jcss.2011.05.002},
  doi          = {10.1016/J.JCSS.2011.05.002},
  bibsource    = {dblp computer science bibliography, https://dblp.org}
}

@inproceedings{dAHK99,
  author       = {Luca de Alfaro},
  editor       = {Jos C. M. Baeten and
                  Sjouke Mauw},
  title        = {Computing Minimum and Maximum Reachability Times in Probabilistic
                  Systems},
  booktitle    = {{CONCUR} '99: Concurrency Theory, 10th International Conference, Eindhoven,
                  The Netherlands, August 24-27, 1999, Proceedings},
  series       = {Lecture Notes in Computer Science},
  volume       = {1664},
  pages        = {66--81},
  publisher    = {Springer},
  year         = {1999},
  url          = {https://doi.org/10.1007/3-540-48320-9\_7},
  doi          = {10.1007/3-540-48320-9\_7},
  bibsource    = {dblp computer science bibliography, https://dblp.org}
}

@inproceedings{dAH00,
  author       = {Luca de Alfaro and
                  Thomas A. Henzinger},
  title        = {Concurrent Omega-Regular Games},
  booktitle    = {15th Annual {IEEE} Symposium on Logic in Computer Science, Santa Barbara,
                  California, USA, June 26-29, 2000},
  pages        = {141--154},
  publisher    = {{IEEE} Computer Society},
  year         = {2000},
  url          = {https://doi.org/10.1109/LICS.2000.855763},
  doi          = {10.1109/LICS.2000.855763},
  bibsource    = {dblp computer science bibliography, https://dblp.org}
}

@article{Martin98,
  author       = {Donald A. Martin},
  title        = {The Determinacy of Blackwell Games},
  journal      = {J. Symb. Log.},
  volume       = {63},
  number       = {4},
  pages        = {1565--1581},
  year         = {1998},
  url          = {https://doi.org/10.2307/2586667},
  doi          = {10.2307/2586667},
  bibsource    = {dblp computer science bibliography, https://dblp.org}
}

@article{MultiObjectiveMDP,
  author       = {Kousha Etessami and
                  Marta Z. Kwiatkowska and
                  Moshe Y. Vardi and
                  Mihalis Yannakakis},
  title        = {Multi-Objective Model Checking of Markov Decision Processes},
  journal      = {Log. Methods Comput. Sci.},
  volume       = {4},
  number       = {4},
  year         = {2008},
  url          = {https://doi.org/10.2168/LMCS-4(4:8)2008},
  doi          = {10.2168/LMCS-4(4:8)2008},
  bibsource    = {dblp computer science bibliography, https://dblp.org}
}

@article{PercentileQueriesMDP,
  author       = {Mickael Randour and
                  Jean{-}Fran{\c{c}}ois Raskin and
                  Ocan Sankur},
  title        = {Percentile queries in multi-dimensional Markov decision processes},
  journal      = {Formal Methods Syst. Des.},
  volume       = {50},
  number       = {2-3},
  pages        = {207--248},
  year         = {2017},
  url          = {https://doi.org/10.1007/s10703-016-0262-7},
  doi          = {10.1007/S10703-016-0262-7},
  bibsource    = {dblp computer science bibliography, https://dblp.org}
}

@article{DisjunctionSG,
  author       = {Tobias Winkler and
                  Maximilian Weininger},
  editor       = {Pierre Ganty and
                  Davide Bresolin},
  title        = {Stochastic Games with Disjunctions of Multiple Objectives},
  booktitle    = {Proceedings 12th International Symposium on Games, Automata, Logics,
                  and Formal Verification, GandALF 2021, Padua, Italy, 20-22 September
                  2021},
  series       = {{EPTCS}},
  volume       = {346},
  pages        = {83--100},
  year         = {2021},
  url          = {https://doi.org/10.4204/EPTCS.346.6},
  doi          = {10.4204/EPTCS.346.6},
  bibsource    = {dblp computer science bibliography, https://dblp.org}
}

@InProceedings{OnStochasticGames,
 author       = {Taolue Chen and
                  Vojtech Forejt and
                  Marta Z. Kwiatkowska and
                  Aistis Simaitis and
                  Clemens Wiltsche},
  editor       = {Krishnendu Chatterjee and
                  Jir{\'{\i}} Sgall},
  title        = {On Stochastic Games with Multiple Objectives},
  booktitle    = {Mathematical Foundations of Computer Science 2013 - 38th International
                  Symposium, {MFCS} 2013, Klosterneuburg, Austria, August 26-30, 2013.
                  Proceedings},
  series       = {Lecture Notes in Computer Science},
  volume       = {8087},
  pages        = {266--277},
  publisher    = {Springer},
  year         = {2013},
  url          = {https://doi.org/10.1007/978-3-642-40313-2\_25},
  doi          = {10.1007/978-3-642-40313-2\_25},
  bibsource    = {dblp computer science bibliography, https://dblp.org}
}

@article{DecidabilitySG,
  author       = {Romain Brenguier and
                  Vojtech Forejt},
  editor       = {Cyrille Artho and
                  Axel Legay and
                  Doron Peled},
  title        = {Decidability Results for Multi-objective Stochastic Games},
  booktitle    = {Automated Technology for Verification and Analysis - 14th International
                  Symposium, {ATVA} 2016, Chiba, Japan, October 17-20, 2016, Proceedings},
  series       = {Lecture Notes in Computer Science},
  volume       = {9938},
  pages        = {227--243},
  year         = {2016},
  url          = {https://doi.org/10.1007/978-3-319-46520-3\_15},
  doi          = {10.1007/978-3-319-46520-3\_15},
  bibsource    = {dblp computer science bibliography, https://dblp.org}
}

@article{HenkinDQBF,
title = {Henkin quantifiers and Boolean formulae: A certification perspective o\
f DQBF},
journal = {Theoretical Computer Science},
volume = {523},
pages = {86-100},
year = {2014},
issn = {0304-3975},
doi = {https://doi.org/10.1016/j.tcs.2013.12.020},
url = {https://www.sciencedirect.com/science/article/pii/S0304397513009328},
author = {Valeriy Balabanov and Hui-Ju Katherine Chiang and Jie-Hong R. Jiang},
}

@misc{StochasticLossyChannel,
      author       = {Daniel Stan and
                  Muhammad Najib and
                  Anthony Widjaja Lin and
                  Parosh Aziz Abdulla},
  editor       = {Aniello Murano and
                  Alexandra Silva},
  title        = {Concurrent Stochastic Lossy Channel Games},
  booktitle    = {32nd {EACSL} Annual Conference on Computer Science Logic, {CSL} 2024,
                  February 19-23, 2024, Naples, Italy},
  series       = {LIPIcs},
  volume       = {288},
  pages        = {46:1--46:19},
  publisher    = {Schloss Dagstuhl - Leibniz-Zentrum f{\"{u}}r Informatik},
  year         = {2024},
  url          = {https://doi.org/10.4230/LIPIcs.CSL.2024.46},
  doi          = {10.4230/LIPICS.CSL.2024.46},
  bibsource    = {dblp computer science bibliography, https://dblp.org} 
}

@inproceedings{ashok2020approximating,
  title={Approximating values of generalized-reachability stochastic games},
  author={Ashok, Pranav and Chatterjee, Krishnendu and K{\v{r}}et{\'\i}nsk{\`y}, Jan and Weininger, Maximilian and Winkler, Tobias},
  booktitle={Proceedings of the 35th Annual ACM/IEEE Symposium on Logic in Computer Science},
  pages={102--115},
  year={2020}
}

@article{fijalkow2010surprizing,
  TITLE = {{The surprizing complexity of generalized reachability games}},
  AUTHOR = {Fijalkow, Nathana{\"e}l and Horn, Florian},
  URL = {https://hal.science/hal-00525762},
  NOTE = {working paper or preprint},
  YEAR = {2010},
  MONTH = Oct,
  HAL_ID = {hal-00525762},
  HAL_VERSION = {v2},
}

@article{Shapley,
author = {L. S. Shapley },
title = {Stochastic Games},
journal = {Proceedings of the National Academy of Sciences},
volume = {39},
number = {10},
pages = {1095-1100},
year = {1953},
doi = {10.1073/pnas.39.10.1095},
URL = {https://www.pnas.org/doi/abs/10.1073/pnas.39.10.1095},
eprint = {https://www.pnas.org/doi/pdf/10.1073/pnas.39.10.1095}}

@article{MixingMDP,
  author       = {Rapha{\"{e}}l Berthon and
                  Shibashis Guha and
                  Jean{-}Fran{\c{c}}ois Raskin},
  editor       = {Holger Hermanns and
                  Lijun Zhang and
                  Naoki Kobayashi and
                  Dale Miller},
  title        = {Mixing Probabilistic and non-Probabilistic Objectives in Markov Decision
                  Processes},
  booktitle    = {{LICS} '20: 35th Annual {ACM/IEEE} Symposium on Logic in Computer
                  Science, Saarbr{\"{u}}cken, Germany, July 8-11, 2020},
  pages        = {195--208},
  publisher    = {{ACM}},
  year         = {2020},
  url          = {https://doi.org/10.1145/3373718.3394805},
  doi          = {10.1145/3373718.3394805},
  bibsource    = {dblp computer science bibliography, https://dblp.org}
}

@inproceedings{chen2012playing,
  title={Playing stochastic games precisely},
  author={Chen, Taolue and Forejt, Vojt{\v{e}}ch and Kwiatkowska, Marta and Simaitis, Aistis and Trivedi, Ashutosh and Ummels, Michael},
  booktitle={International Conference on Concurrency Theory},
  pages={348--363},
  year={2012},
  organization={Springer}
}

@inproceedings{kr21,
  author       = {Julian Gutierrez and
                  Lewis Hammond and
                  Anthony W. Lin and
                  Muhammad Najib and
                  Michael J. Wooldridge},
  title        = {Rational Verification for Probabilistic Systems},
  booktitle    = {{KR}},
  pages        = {312--322},
  year         = {2021}
}

@inproceedings{prism-games2,
  author       = {Marta Kwiatkowska and
                  David Parker and
                  Clemens Wiltsche},
  editor       = {Marsha Chechik and
                  Jean{-}Fran{\c{c}}ois Raskin},
  title        = {PRISM-Games 2.0: {A} Tool for Multi-objective Strategy Synthesis for
                  Stochastic Games},
  booktitle    = {Tools and Algorithms for the Construction and Analysis of Systems
                  - 22nd International Conference, {TACAS} 2016, Held as Part of the
                  European Joint Conferences on Theory and Practice of Software, {ETAPS}
                  2016, Eindhoven, The Netherlands, April 2-8, 2016, Proceedings},
  series       = {Lecture Notes in Computer Science},
  volume       = {9636},
  pages        = {560--566},
  publisher    = {Springer},
  year         = {2016},
  url          = {https://doi.org/10.1007/978-3-662-49674-9\_35},
  doi          = {10.1007/978-3-662-49674-9\_35},
  bibsource    = {dblp computer science bibliography, https://dblp.org}
}

@inproceedings{LR16,
  author       = {Anthony W. Lin and
                  Philipp R{\"{u}}mmer},
  title        = {Liveness of Randomised Parameterised Systems under Arbitrary Schedulers},
  booktitle    = {{CAV} {(2)}},
  series       = {Lecture Notes in Computer Science},
  volume       = {9780},
  pages        = {112--133},
  publisher    = {Springer},
  year         = {2016}
}

@inproceedings{LLMR17,
  author       = {Ondrej Leng{\'{a}}l and
                  Anthony Widjaja Lin and
                  Rupak Majumdar and
                  Philipp R{\"{u}}mmer},
  title        = {Fair Termination for Parameterized Probabilistic Concurrent Systems},
  booktitle    = {{TACAS} {(1)}},
  series       = {Lecture Notes in Computer Science},
  volume       = {10205},
  pages        = {499--517},
  year         = {2017}
}

@book{Lynch-book,
    title = {Distributed Algorithms},
    author = {Nancy Lynch},
    year = {1996},
    publisher = {Morgan Kaufmann}
}

@book{Fokkink-book,
    title = {Distributed Algorithms},
    author = {Wan Fokkink},
    year = {2013},
    publisher={MIT Press}
}

@article{MMSZ20,
  author       = {Rupak Majumdar and
                  Kaushik Mallik and
                  Anne{-}Kathrin Schmuck and
                  Damien Zufferey},
  title        = {Assume-Guarantee Distributed Synthesis},
  journal      = {{IEEE} Trans. Comput. Aided Des. Integr. Circuits Syst.},
  volume       = {39},
  number       = {11},
  pages        = {3215--3226},
  year         = {2020},
  url          = {https://doi.org/10.1109/TCAD.2020.3012641},
  doi          = {10.1109/TCAD.2020.3012641},
  bibsource    = {dblp computer science bibliography, https://dblp.org}
}

\longshort{
\appendix

\clearpage
\balance
\section{Appendix}
\subsection{Proof of Lemma~\ref{lem:conj-as-reach}}
    We show via induction that the lemma holds for any conjunction of up to $n$ objectives.
    For $n=1$, $\phi' = \phi = AS(\lozenge T_1)$, which is a single almost sure reachability query and therefore determined.
    
    Assume now that the lemma holds for any conjunction with up to $k$ objectives, and consider
    a query $\phi = \bigwedge_{i = 1}^{k+1} AS(\lozenge T_i)$.
    Since $AS(\lozenge T')$ is determined, we consider two cases: either player 1 has a winning strategy $\sigma'$, or player 2 has a winning strategy $\tau'$. We show that in both cases, the same player has a winning strategy for $\phi$.\\
    
    Assume $\sigma'$ is a winning strategy for $AS(\lozenge T')$ and let $\tau$ be any strategy of player 2.
    From the construction of $T'$ it follows that for every $i\in \{1, \dots, k+1\}$, $s\in T'\cap T_i$, there exists a strategy $\sigma_s$ such that
    $\sigma_s, \tau \models_{\mathcal{G}_s} \bigwedge_{j\neq i}AS(\lozenge T_j)$. Let $\sigma$ be the strategy that plays $\sigma'$ until some $x\in T'$
    is reached and then switches to playing $\sigma_s$, and for $s\in T'$ let $\lozenge_1 s$ denote the plays where $s$ is reached before any other $s'\in T', s'\neq s$.
    It follows that for $i\in \{1, \dots , k+1\}$:
    \begin{flalign*}
        \Pr^{\sigma, \tau}(\lozenge T_i) & \geq \Pr^{\sigma, \tau}(\lozenge T_i | \lozenge T') \\
        & = \sum_{s\in T'} \Pr^{\sigma, \tau}(\lozenge_1 s) \cdot \Pr^{\sigma, \tau}(\lozenge T_i | \lozenge_1 s) & \\
        & = \sum_{s\in T'\cap T_i} \Pr^{\sigma, \tau}(\lozenge_1 s) + \sum_{s\in T'\setminus T_i} \Pr^{\sigma, \tau}(\lozenge_1 s) \cdot \Pr^{\sigma, \tau}(\lozenge T_i | \lozenge_1 s) & \\
        & = \sum_{s\in T'\cap T_i} \Pr^{\sigma', \tau}(\lozenge_1 s) + \sum_{s\in T'\setminus T_i} \Pr^{\sigma', \tau}(\lozenge_1 s) \cdot \Pr^{\sigma_s, \tau}_{\mathcal{G}_s}(\lozenge T_i) & \\
        & = \Pr^{\sigma', \tau}(\lozenge T') = 1 &
    \end{flalign*}
    Therefore $\sigma$ is a winning strategy for $\phi$.

    Assume now instead that player 2 has a winning strategy $\tau'$ for $AS(\lozenge T')$ and let $\sigma$ be any strategy of player 1.
    From the induction hypothesis, it follows that each query of the the form $\bigwedge_{j\neq i}AS(\lozenge T_j)$ is determined, as it is a conjunction of at most $k$ objectives.
    Together with the construction of $T'$ it follows that for every $s\in T_i\setminus T'$, there is a strategy $\tau_s$ such that $\sigma, \tau_s \not\models_{\mathcal{G}_s} AS(\lozenge T_j)$ for some $j\neq i$. 
    Let $\tau$ be the strategy that plays $\tau'$
    until some $s\in T_i\setminus T'$ is reached and then switches to playing $\tau_s$.
    Let $T = \bigcup T_i$. If $\Pr^{\sigma, \tau}(\lozenge T) < 1$, then it follows that $\sigma, \tau\not\models \phi$.
    Assume otherwise $\Pr^{\sigma, \tau}(\lozenge T) = 1$ and let $\lozenge_1 s$ denote the plays where $s\in T$ is reached before any other $s'\in T$, $s'\neq s$.
    If $\forall s\in T$ with $\Pr^{\sigma, \tau'}(\lozenge_1 s) > 0$ it holds that $s\in T'$, then it follows that $\Pr^{\sigma, \tau'}(\lozenge T') = 1$, which is a contradiction to $\tau'$ being a winning strategy for $AS(\lozenge T')$ for player 2. Therefore there exists $j\in \{1, \dots k+1\}$, $s\in T_j\setminus T'$, with $\Pr^{\sigma, \tau'}(\lozenge_1 s) > 0$. It follows that:
    \begin{flalign*}
        \Pr^{\sigma, \tau}(\lnot \lozenge T_j) & \geq \Pr^{\sigma, \tau}(\lnot \lozenge T_j | \lozenge_1 s) \cdot \Pr^{\sigma, \tau}(\lozenge_1 s) \\
         & = \Pr^{\sigma, \tau_s}_{\mathcal{G}_s}(\lnot \lozenge T_j) \Pr^{\sigma, \tau'}(\lozenge_1 s) > 0 &
    \end{flalign*}
    Therefore $\tau$ is a winning strategy for player 2 for $\phi$.
    
    It follows that $\phi$ is determined, and player 1 has a winning strategy if and only if they have a winning strategy for $AS(\lozenge T')$. Therefore the lemma holds for $n = k+1$, and
    by induction it holds for all $n\in \mathbb{N}$.
\subsection{Algorithm for Lemma~\ref{lem:conj-as-one-nz-reach}}
\begin{algorithm}[t]
    \caption{$\mathit{win}((s, b), h)$ \label{algo:win}}
    \begin{algorithmic}
        \\
        \If{$b_i = 1$ for all $i$}
            \State \Return $\mathit{true}$
        \EndIf
        \If{$(s, b) \in h$}
            \State \Return $\mathit{false}$
        \EndIf
        \\
        \If{$s\in S_1$}
            \State \Return $\exists s'\in A(s): \mathit{win}((s', b'), h\cup \{(s, b)\})$
        \ElsIf{$s\in S_2$}
            \State \Return $\forall s'\in A(s): \mathit{win}((s', b'), h\cup \{(s, b)\})$
        \ElsIf{$s\in S_c$}
            \If{$s$ is in the winning region of $\mathcal{G}^*$ for $\bigwedge_{i\in I} AS(\lozenge T_i)$}
                \State \Return $\exists s'\in A(s): \mathit{win}((s', b'), h\cup \{(s, b)\})$
            \Else
                \State \Return $\mathit{false}$
            \EndIf
        \EndIf
    \end{algorithmic}
\end{algorithm}

Algorithm 1 shows the algorithm described in Lemma~\ref{lem:conj-as-one-nz-reach}. This algorithm checks if there is a tree in the goal unfolding of $\mathcal{G}$ which branches universally at states of player 2 where every leaf is in $T_0'$ and each node of the tree is in the winning region of $\bigwedge_{i\in I} AS(\lozenge T_i)$. While the goal unfolding is exponential in size, the length of any simple path is at most linear in the size of the original state space $|S|$. Since the algorithm terminates on a branch after encountering a loop, it requires only polynomial space.
\subsection{Proof of Lemma~\ref{lem:as-safe}}
\begin{proof}    
    First note that by construction of $\mathcal{G'}$, it follows that $\forall s\in S_2': s_\bot\not\in A(s)$ and $\forall s\in S_c': P(s)(s_\bot) = 0$.
    This is because, otherwise, player 1 does not win $AS(\square T)$ from $s$. Therefore a play of $\mathcal{G'}$ only reaches $s_\bot$ after a state $s\in S_1'$.
    
    Assume player 1 has a winning strategy $\sigma$ for $\psi$ in $\mathcal{G'}$. Let $\sigma'$ be the strategy of player 1 that plays according to $\sigma$ except at any path $\pi s$
    where $\sigma(\pi s)(s_\bot) > 0$, where 
    instead of playing $s_\bot$ the strategy switches to playing $\sigma_s$, the winning strategy for $AS(\square T)$ in $\mathcal{G}_s$.
    Since $s_\bot$ is not in any of the reachability targets of $\psi$, and no play of the strategy $\sigma'$ leaves $\Sem{AS(\square T)}{1}$,    
    it follows that $\sigma'$ is a winning strategy for $\phi$ in $\mathcal{G}$.
    
    Assume now player 2 has a winning strategy $\tau$ for $\psi$ in $\mathcal{G'}$ and for each $s\not\in \Sem{AS(\square T)}{1}$ let $\tau_s$ be the winning strategy
    of player 2 for $AS(\square T)$. The strategy $\tau'$ that plays $\tau$ until some $s\not\in \Sem{AS(\square T)}{1}$ is reached is then a winning strategy for player 2 for $\phi$:
    For any strategy $\sigma$ of player 1, either $\sigma, \tau' \models_{\mathcal{G}} AS(\square T)$,
    in which case $\sigma, \tau' \not\models_{\mathcal{G}} \phi$ because $\sigma, \tau\not\models_{\mathcal{G}'} \psi$,
    or $\sigma, \tau' \not\models_{\mathcal{G}} AS(\square T)$ and therefore $\sigma, \tau' \not\models \phi$.
    
    It follows that $\phi$ is determined, and player 1 has a winning strategy for $\phi$ in $\mathcal{G}$ if and only if player 1 has a winning strategy for $\psi$ in $\mathcal{G}'$.
\end{proof}

\subsection{Proof of Theorem~\ref{lem:hardness}}

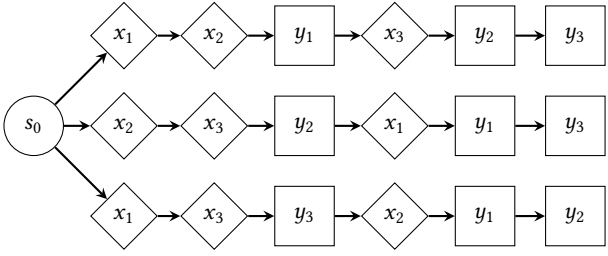
\begin{figure}[t]
    \centering    
    \begin{tikzpicture}[scale=0.4, node distance = 1.2cm, on grid, auto]
        \node (b11) [state, diamond] {$x_1$};        
        \node (b21) [state, below = of b11, diamond] {$x_2$};
        \node (b31) [state, below = of b21, diamond] {$x_1$};
        \node (s0) [state, left = of b21] {$s_0$};
        
        \node (b12) [state, right = of b11, diamond] {$x_2$};
        \node (b13) [state, right = of b12, rectangle] {$y_1$};
        \node (b14) [state, right = of b13, diamond] {$x_3$};
        \node (b15) [state, right = of b14, rectangle] {$y_2$};
        \node (b16) [state, right = of b15, rectangle] {$y_3$};
        
        \node (b22) [state, right = of b21, diamond] {$x_3$};
        \node (b23) [state, right = of b22, rectangle] {$y_2$};
        \node (b24) [state, right = of b23, diamond] {$x_1$};
        \node (b25) [state, right = of b24, rectangle] {$y_1$};
        \node (b26) [state, right = of b25, rectangle] {$y_3$};

        \node (b32) [state, right = of b31, diamond] {$x_3$};
        \node (b33) [state, right = of b32, rectangle] {$y_3$};
        \node (b34) [state, right = of b33, diamond] {$x_2$};
        \node (b35) [state, right = of b34, rectangle] {$y_1$};
        \node (b36) [state, right = of b35, rectangle] {$y_2$};

        \path[-stealth, thick]
            (s0) edge (b11) 
            (s0) edge (b21) 
            (s0) edge (b31)
            
            (b11) edge (b12)
            (b12) edge (b13)
            (b13) edge (b14)
            (b14) edge (b15)
            (b15) edge (b16)

            (b21) edge (b22)
            (b22) edge (b23)
            (b23) edge (b24)
            (b24) edge (b25)
            (b25) edge (b26)

            (b31) edge (b32)
            (b32) edge (b33)
            (b33) edge (b34)
            (b34) edge (b35)
            (b35) edge (b36);            
    \end{tikzpicture}
\caption{Example construction for a DQBF formula
        $\Phi = \forall x_1, x_2, x_3 \exists y_{1, S_1}\exists y_{2, S_2}\exists y_{3, S_3} \phi$
        where $S_1 = \{x_1, x_2\}$, $S_2 = \{x_2, x_3\}$, $S_3 = \{x_1, x_3\}$. Each node annotated with a variable represents the module for that variable.}
    \label{fig:dqbfconstruction}
\end{figure}

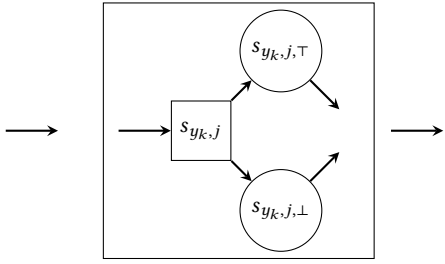
\begin{figure}[t]
    \begin{tikzpicture}[scale=0.4, node distance = 1.5cm, on grid, auto]
        
        \node (s0) [state, draw=none] {};   
        \node (s1) [state, right = of s0, rectangle] {$s_{y_k, j}$};
        \node (s2) [state, above right = of s1] {$s_{y_k, j, \top}$};
        \node (s3) [state, below right = of s1] {$s_{y_k, j, \bot}$};
        \node (s4) [state, above right = of s3, draw=none] {};
        \node (s5) [state, right = of s4, draw=none] {};
        \node (s6) [state, left = of s0, draw=none] {};

        \draw (0.5,-4.25) rectangle (9.5,4.25);
        \path[-stealth, thick]
            (s0) edge (s1) 
            (s1) edge (s2) 
            (s1) edge (s3)
            (s2) edge (s4)   
            (s3) edge (s4)    
            (s4) edge (s5)  
            (s6) edge (s0);
    \end{tikzpicture}
    \caption{Construction for the module of a variable $y_k$ on branch $j$. In a module for a variable $x_k$, the node $s_{x_k, j}$ is instead controlled by player 2.}
    \label{fig:dqbfmodule}
\end{figure}
Let $\Phi$ be a DQBF in S-Form.
$$ \Phi = \forall x_1, \dots , \forall x_n \exists y_{1, S_1} , \dots, \exists y_{m, S_m} \phi$$
We construct a stochastic game $\mathcal{G}_\Phi$ where player 1 has a winning strategy if and only if $\Phi$ is satisfied.

The game consists of the following states: The initial state $s_0\in S_c$, which randomly selects one of $m$ branches.
Each branch $j\in \set{1,\ldots, m}$ consists of $n+m$ modules, one for each variable, ordered according to the set $S_j$:
First we have the modules for each $x_i\in S_j$, followed by the module for $y_j$, then the modules for each $x_i\not\in S_j$ and finally all modules for $y_i$ with $i\neq j$.

The module for a variable $v$ in branch $j$ consists of a state $s_{v, j}$ controlled by player 1 if $v\in Y$, or by player 2 if $v\in X$. these states have transitions to states $s_{v, j, \top}$ and $s_{v, j, \bot}$
which each have a single transition to the next module in the branch (or are the leaves of the branch, if they are in the final module). We construct the winning condition in a way such that visiting such a state corresponds to setting the variable to true or false respectively.

The winning condition $\psi$ for $\mathcal{G}_\Phi$ is derived from the formula $\Phi$ as follows:
For a variable $v\in X\cup Y$, let $T_{v, \top} = \{s_{v, j, \top} | j\in \{1, \dots, m\}\}$ and $T_{v, \bot} = \{s_{v, j, \bot} | j\in \{1, \dots, m\}\}$.
In $\phi$, we first move all negation to the variables and then replace each positive occurrence of $v$ in $\phi$ with $AS(\lozenge T_{v, \top})$ and each negative occurrence with $AS(\lozenge T_{v, \bot})$ to obtain the formula $\phi'$.
Additionally, we force both players to play deterministically and identically in each branch, by adding restrictions $\psi_1$, $\psi_2$ for each player:
$$\psi_1 = \bigwedge_{y\in Y} AS(\lozenge T_{y, \top}) \lor AS(\lozenge T_{y, \bot})$$
$$\psi_2 = \bigvee_{x\in X} NZ(\lozenge T_{x, \top}) \land NZ(\lozenge T_{x, \bot})$$
Player 1 wins the game by achieving $\phi'$ while satisfying $\psi_1$, or if player 2s strategy satisfies $\psi_2$ (by not playing deterministically and identically):
$$\psi = ( \phi' \land \psi_1)  \lor \psi_2 $$
$\psi$ is now a positive Boolean combination of almost sure and nonzero reachability queries.

We now show that there is a Skolemization that satisfies $\Phi$ if and only if player 1 has a winning strategy for $\psi$ in $\mathcal{G}_\Phi$.
For this note that a strategy $\sigma$ of player 1 that satisfies $\psi_1$ (by playing deterministically and identically on each branch) can be translated to a set of Skolem functions $Y_j: S_j \to \{0, 1\}$ and vice versa:
$$Y_i(\alpha) = 1 \iff \sigma(\pi s_{y_j, j})(s_{v, j, \top}) = 1$$
where $\alpha$ is the assignment to $x_i\in S_j$ such that $\alpha(x_i) = 1$ iff $s_{x_i, j, \top}\in \pi$.
At $s_{y_j, j}$ on branch $j$, the strategy $\sigma$ chooses to play either $s_{y_j, j, \top}$ or $s_{y_j, j, \bot}$ with probability 1, based on the current history of the play up to $s_{y_j, j}$, which includes exactly the choices made by player 2 for all $x_i\in S_j$.
The choices made by $\sigma$ at $s_{y_j, i}$, $j\neq i$, can be defined identically, since on branch $i\neq j$, all modules for variables $x_i\in S_j$ are seen before $s_{y_j, i}$.

We show that player 1 has a winning strategy for $\psi$ in $\mathcal{G}_\Phi$ if and only if $\Phi$ is satisfied.

Assume $\sigma$ is a winning strategy for $\psi$ in $\mathcal{G}_\Phi$
and let $\tau'$ be any deterministic strategy of player 2 that plays identically at each branch. Since $\sigma$ is a winning strategy and $\sigma, \tau' \not\models \psi_2$, it follows that $\sigma, \tau' \models \psi_1$ and therefore $\sigma$ is a deterministic strategy that plays identically at each branch.
Let $Y_j$ be the Skolem functions constructed from $\sigma$, let $\alpha$ be any assignment to the variables $X$ and let $\tau$ be the strategy that plays $s_{x, j, \top} \iff \alpha(x_i) = 1$. Now $\tau$ is a deterministic strategy that plays identically at each branch and therefore $\sigma, \tau \not\models \psi_2$. It follows that $\sigma, \tau \models \phi'$ and therefore $\phi[x_i \to \alpha, y_j \to Y_j(\alpha)]$ is satisfied.
Therefore $\forall x_1, \dots x_n: \phi[y_j \to Y_j(X)]$ is satisfied, and $\Phi$ is satisfied.

Assume now $\Phi$ is satisfied. Then there exist Skolem functions $Y_j$ such that $\forall x_1, \dots x_n: \phi[y_j \to Y_j(X)]$ is satisfied.
Let $\sigma$ be the strategy constructed from these Skolem functions.
Let $\tau$ be any strategy of player 2. If $\tau$ is not a deterministic strategy that plays identically on each branch, then it follows that $\sigma, \tau \models \psi_2$ and therefore $\sigma, \tau \models \phi$.
Otherwise assume $\tau$ is a deterministic strategy that plays identically on each branch. It follows that $\sigma, \tau \models \psi_1$ by construction of $\sigma$ from $Y_j$.
Let $\alpha$ be the assignment such that $\alpha(x_i) = 1$ iff $\tau$ plays $s_{x, j, \top}$ with probability 1 ($\forall j$).
It follows that $\phi[x_i\to \alpha, y_i \to Y_i(\alpha)]$ is satisfied and therefore $\sigma, \tau \models \phi'$ and $\sigma, \tau \models \psi$. Therefore $\sigma$ is a winning strategy for $\psi$ in $\mathcal{G}_\Phi$. \hfill$\square$.

\subsection{Proof of Lemma~\ref{lem:bsigma_win}}
We first consider a query $\phi$ that is a positive Boolean combination of $AS\lozenge$ and $NZ\lozenge$, and later discuss how this is generalized to include $AS\square$ as well.

Overview:
We assume that $\bsigma$ is not a winning strategy and let $\tau$ be a strategy of player 2 such that $\bsigma, \tau \not\models \phi$. We first show that then there also exists a strategy $\btau$ that uses the same memory as $\bsigma$ such that $\bsigma, \btau \not\models \phi$.
We then show that $\sigma, \btau \models NZ(\lozenge T) \implies \bsigma, \btau \models NZ(\lozenge T)$. It follows that there are some almost sure queries $AS(\lozenge T_i)$ in $\phi$ such that $\bsigma, \btau \not\models AS(\lozenge T_i)$. From this we construct a strategy $\tau^*$ with $\sigma, \tau^* \not\models AS(\lozenge T_i)$ for any almost sure query not satisfied by $\bsigma, \btau$.
We then show that $\sigma, \tau^* \models NZ(\lozenge T) \implies \bsigma, \btau \models NZ(\lozenge T)$, which implies that $\sigma, \tau^* \not\models \phi$. This is a contradiction to $\sigma$ being a winning strategy, therefore $\bsigma$ must be a winning strategy.

Assume $\bsigma$ is not a winning strategy and let $\tau$ be a strategy of player 2 such that $\bsigma, \tau \not\models \phi$.

It follows that $\tau \models_{\mathcal{G}_{\bsigma}} \lnot \phi$, where $\mathcal{G}_{\bsigma}$ is the MDP induced by $\bsigma$. Since $\bsigma$ only requires exponential memory, $\mathcal{G}_{\bsigma}$ is finite
and there exists a strategy $\btau$ with $\btau \models_{\mathcal{G}_{\bsigma}} \lnot \phi$ that only uses the set of previously visited states as memory in $\mathcal{G}$ \cite{MultiObjectiveMDP}.
It follows that $\bsigma, \btau \not\models \phi$.

Consider that by construction of $\bsigma$:
\begin{align*}
    \sigma(\pi s)(s') > 0 \implies \bsigma(\pi s)(s') > 0\ \ \ \  (**)
\end{align*}
Therefore for any finite path $\pi$ with $\Pr^{\sigma, \tau}(\pi) > 0$, it follows that  $\Pr^{\bsigma, \tau}(\pi) > 0$.
It follows that $\sigma, \btau \models NZ(\lozenge T) \implies \bsigma, \btau \models NZ(\lozenge T)$ for any target set $T$.

Since $\sigma$ is a winning strategy, $\sigma, \btau \models \phi$. Since however $\bsigma, \btau \not\models \phi$, it follows that there are some almost sure targets
$AS(\lozenge T)$ in $\phi$ with $\bsigma, \btau\not\models AS(\lozenge T)$.
Let $I^* = \{ i| \bsigma, \btau \not\models AS(\lozenge T_i)\} \neq \emptyset$ and consider some target $i\in I^*$.

Since $\bsigma$ and $\btau$ only use exponential memory, the induced Markov Chain $\mathcal{M}_\mathcal{G}^{\bsigma, \btau}$ is finite.
It follows that for each $i\in I^*$, there is a path $\pi_i s_i$ such that
$\Pr^{\bsigma, \btau}(\pi_i s_i) > 0$ and $\Pr^{\bsigma, \btau}(\lozenge T_i | \pi_i s_i) = 0$.
Let $\tau_i', \pi_i'$ such that $\Pr^{\sigma, \tau_i'}(\pi_i' s_i) > 0$ and $set(\pi_i) = set(\pi_i')$, which exists due to $(*)$.

Let then $\tau_i^*$ be the strategy that on
any prefix of $\pi_i's_i$ plays the next action in the path with probability $\frac{1}{2}$ and with probability $\frac{1}{2}$ plays $\btau$, and plays $\btau$ everywhere else.
It follows that $\Pr^{\sigma, \tau_i^*}(\pi_i' s_i) > 0$ and $\Pr^{\sigma, \tau_i^*}(\lozenge T_i | \pi_i' s_i) = \Pr^{\sigma, \btau}(\lozenge T_i | \pi_i' s_i) = 0$
since $\Pr^{\bsigma, \btau}(\lozenge T_i | \pi_i s_i) = 0$ due to $(**)$, and $\bsigma, \btau$ have the same memory at $\pi_i s_i$ and $\pi_i's_i$.
Therefore $\sigma, \tau_i^* \not\models AS(\lozenge T_i)$.

Let $\tau^*$ be the strategy that plays each $\tau_i$ with probability $\frac{1}{|I^*|}$.
Then it follows that $\sigma, \tau^* \not\models AS(\lozenge T_i)$ for all $i\in I^*$, and therefore $\sigma, \tau^* \models AS(\lozenge T_i) \implies \bsigma, \btau \models AS(\lozenge T_i)$.

Consider an objective $NZ(\lozenge T)$ with $\sigma, \tau^* \models NZ(\lozenge T)$. Then there exists a path $\pi\in \lozenge T$ with $\Pr^{\sigma, \tau^*}(\pi) > 0$
and therefore $\Pr^{\sigma, \tau^*_i}(\pi) > 0$ for some $i$.
If $\pi_i's_i$ is a prefix of $\pi$, then let $\pi = \pi_i' s_i \pi_2$ and $\pi^* = \pi_i s_i \pi_2$.
Since $\pi \in \lozenge T$ and $set(\pi_i) = set(\pi_i')$, it follows that $\pi^* \in \lozenge T$.
Consider that since $\Pr^{\sigma, \tau_i^*}(\pi_i' s_i \pi_2 | \pi_i' s_i) > 0$, and $\tau_i^* = \btau$ everywhere except on $\pi_i' s_i$ it follows that $\Pr^{\bsigma, \btau}(\pi_i s_i \pi_2 | \pi_i s_i) > 0$, because both $\bsigma$ and $\btau$ have the same memory after $\pi_i$ and $\pi_i'$.
It follows that $\Pr^{\bsigma, \btau}(\pi^*) = \Pr^{\bsigma, \btau}(\pi_i s_i) \Pr^{\bsigma, \btau}(\pi_i s_i \pi_2 | \pi_i s_i) > 0$ and therefore $\bsigma, \btau \models NZ(\lozenge T)$.

We now have that $\sigma, \tau^* \models NZ(\lozenge T_i) \implies \bsigma, \btau \models NZ(\lozenge T_i)$ and $\sigma, \tau^* \models AS(\lozenge T_i) \implies \bsigma, \btau \models AS(\lozenge T_i)$ for any target sets in $\phi$.
Since $\phi$ is a positive Boolean formula and $\bsigma, \btau \not\models \phi$, it follows that $\sigma, \tau^* \not\models \phi$ which is a contradiction to $\sigma$ being a winning strategy.

Therefore the assumption that $\bsigma$ is not a winning strategy is false, and $\bsigma$ is a winning strategy.

We note that the construction can easily be extended to also include almost sure safety queries, by letting $I^* = I^*_\lozenge \cup I^*_\square$ to also include $AS(\square T_i)$ targets that are not satisfied by $\bsigma, \btau$.
For $i\in I^*_\square$, $\tau_i^*$ is then the strategy such that $\sigma, \tau_i^*$ reaches $s_i\not\in T_i$ with positive probability through a path $\pi_i s_i$, which exists due to $(**)$.

We then construct $\tau^*$ as a randomization over strategies that spoil almost sure reachability and those that spoil almost sure safety targets, and both
$\sigma, \tau^* \models AS(\lozenge T_i) \implies \bsigma, \btau \models AS(\lozenge T_i)$
and
$\sigma, \tau^* \models AS(\square T_i) \implies \bsigma, \btau \models AS(\square T_i)$
follow as before.
}
{}

\end{document}